\documentclass[journal]{IEEEtran}

\usepackage{booktabs} % For formal tables
\usepackage{tabularx}
\usepackage[strict]{changepage}
\usepackage{calc}
\usepackage{mathtools,lipsum}
\usepackage{amsmath}
\usepackage{amssymb}
\usepackage{mathrsfs}
\usepackage{upgreek}
\usepackage{graphicx}
\usepackage{algorithmicx}
\usepackage{algorithm}
\usepackage{enumitem}
\usepackage[noend]{algpseudocode}
\usepackage{caption}
\usepackage{subcaption}
\usepackage{amsthm, times, amsfonts, cite}

\newtheorem{theorem}{Theorem}
\newtheorem{lemma}{Lemma}
\newtheorem{corollary}{Corollary}

\usepackage[T1]{fontenc}
\usepackage[latin9]{inputenc}
\usepackage{array}
\usepackage{float}
\usepackage{amsmath}
\usepackage{amssymb}
\usepackage{graphicx}

\usepackage{amsmath}
\begin{document}
\title{
Optimized Video Streaming over Cloud: A Stall-Quality Trade-off }

\author{Abubakr Alabbasi and Vaneet Aggarwal    \thanks{The authors are affiliated with Purdue University, West Lafayette, IN 47907, email:\{aalabbas,vaneet\}@purdue.edu.}}

% The default list of authors is too long for headers}
\maketitle

%\keywords{Video Streaming over Cloud, Erasure Codes, Mean Stall Duration,  Video Quality, Two-stage Probabilistic Scheduling}

%
% The code below should be generated by the tool at
% http://dl.acm.org/ccs.cfm
% Please copy and paste the code instead of the example below. 
%\ccsdesc[500]{Computer systems organization~Embedded systems}
%\ccsdesc[300]{Computer systems organization~Redundancy}
%\ccsdesc{Computer systems organization~Robotics}
%\ccsdesc[100]{Networks~Network reliability}
% We no longer use \terms command
%\terms{Theory}

% \keywords{ACM proceedings, \LaTeX, text tagging}

\begin{abstract}

As video-streaming services have expanded and improved, cloud-based video has evolved into a necessary feature of any successful business for reaching internal and external audiences. In this paper, video streaming over distributed storage is considered where the video segments are encoded using an erasure code for better reliability. There are multiple parallel streams between each server and the edge router. For each client request, we need to determine the subset of servers to get the data, as well as one of the parallel stream from each chosen server. In order to have this scheduling, this paper proposes a two-stage probabilistic scheduling. The selection of video quality is also chosen with a certain probability distribution. With these parameters, the playback time of video segments is determined by characterizing the download time of each coded chunk for each video segment. Using the playback times, a bound on the moment generating function of the stall duration is used to bound the mean stall duration. Based on this, we  formulate an optimization problem to jointly optimize the convex combination of mean stall duration and average video quality for  all requests, where the two-stage probabilistic scheduling, probabilistic video quality selection, bandwidth split among parallel streams, and auxiliary bound parameters can be chosen. This non-convex problem  is solved using an efficient iterative algorithm.  Evaluation results show significant improvement in QoE metrics for cloud-based video as compared to the considered baselines.

%, and hence, being the first work to the best of our knowledge that considers video streaming over erasure-coded distributed cloud systems. For a selected quality, the playback time of video segments is determined by characterizing the download time of each coded chunk for each video segment. Then,   using the playback times, a bound on the moment generating function on the stall duration is used to bound the mean stall duration.  The weighted average quality of the streamed video and the mean stall duration are crucial  quality of experience (QoE) measures for the end users.  Thus, based on these metrics, we formulate an optimization problem to jointly optimize the convex combination of both the QoE metrics averaged over all requests over the access of the video content, the allocated bandwidth among the different queues and the quality of the streamed quality. The non-convex problem  is solved using an efficient iterative algorithm.  Numerical results show significant improvement in QoE metrics for cloud-based video as compared to the considered baselines.
\end{abstract}

%\begin{abstract}
%The demand for global video has been burgeoning across industries. With the expansion and improvement of  video-streaming services, cloud-based video is evolving into a necessary feature of any successful business for reaching internal and external audiences. In this paper, we consider video streaming over distributed systems where the video segments are encoded using an erasure code for better reliability thus being the first work to our best knowledge that considers video streaming over erasure-coded distributed cloud systems. The download time of each coded chunk of each video segment is characterized and ordered statistics over the choice of the erasure-coded chunks is used to obtain the playback time of different video segments. Using the playback times, bounds on the moment generating function on the stall duration is used to bound the mean stall duration.  The average quality of the streamed video and the mean stall duration  are important quality of experience (QoE) measures for the end users.  Based on these metrics, we formulate optimization problem to jointly optimize the convex combination of both the QoE metrics averaged over all requests over the access of the video content, the allocated bandwidth among the different queues and the quality of the streamed quality. The non-convex problem  is solved using an efficient iterative algorithm.  Numerical results show significant improvement in QoE metrics for cloud-based video as compared to the considered baselines.
%\end{abstract}

\begin{IEEEkeywords}
	Video Streaming over Cloud, Erasure Codes, Mean Stall Duration,  Video Quality, Two-stage probabilistic scheduling. \end{IEEEkeywords}

\section{Introduction}

Cloud computing has changed the way many Internet services are provided and operated. Video-on-Demand (VoD)  providers are increasingly moving their streaming services, data storage, and encoding software to cloud service providers \cite{yCloud2010,Aggarwal:multimedia:2013}. With the annual growth of global video streaming at a rate of 18.3\% \cite{marketsandmarkets},  cloud-based video has become an imperative feature of any successful business. For example,  IBM estimates cloud-based video will be a \$105 billion market opportunity by 2019 \cite{ibm}. In this paper, we will give a novel approach to an optimized cloud-based-video streaming.

%Recently, it becomes feasible for VoD providers to use computing resources of could service providers, instead of buying server racks to build their own data centers. For instance, Netflix moved its streaming services , data storage, and encoding software to Amazon Web Services in 2010 \cite{yCloud2010}. 

%The demands of video streaming services have been skyrocketing over these years, with the global video streaming market expected to grow annually at a rate of  18.3\% \cite{marketsandmarkets}.  With the proliferation and advancement of  video-streaming services, cloud-based video has become an imperative feature of any successful business. This can also be seen as IBM estimates cloud-based video will be a \$105 billion market opportunity by 2019 \cite{ibm}.

%Today, cloud-based video does not use erasure coding. One of the key reason is  the additional decoding latency from multiple coded streams.
 Since the  computing has been growing exponentially \cite{Denning:2016:ELC:3028256.2976758}, 
% it is only a matter of time when 
 the computation of decoding will not limit the latencies in delay sensitive video streaming and the networking latency will govern the system designs. The key advantage of erasure coding is that it  reduces storage cost while providing similar reliability as replicated systems \cite{2015_1, Dimakis:10}, and thus has now been widely adopted by companies like Facebook \cite{Sathiamoorthy13}, Microsoft \cite{Asure14}, and Google \cite{Fikes10}.  Further, we note that replication is a special case of erasure coding. Thus, the proposed research using erasure-coded content on the servers can also be used when the content is replicated on the servers.

%In cloud storage systems, erasure coding emerged as a promising technique to reduce the storage cost for a given reliability as compared to the replicated systems \cite{2015_1,Dimakis:10}. It has been widely adopted in modern storage systems by companies like Facebook \cite{Sathiamoorthy13}, Microsoft \cite{Asure14} and Google \cite{Fikes10}. This paper considers video streaming when the video content is placed on cloud servers, where erasure coding is considered. 

In cloud-based-video, the users are connected to an edge router, which fetch the contents from the distributed storage servers (as depicted in Fig. \ref{fig:sys}). There are multiple parallel streams (PSs) between a server and the edge router which help in getting multiple streams simultaneously. We assume that the connection between users and edge router is not limited.  Unlike the case of file download, the later video-chunks do not have to be downloaded as fast as possible to improve the quality-of-experience (QoE) and thus multiple parallel streams help achieve better QoE. The key QoE metrics for video streaming are the duration of  stalls at the clients and the streamed average video quality. Every viewer can relate the QoE for watching videos to the stall duration and is thus one of the key focus in the studied streaming algorithms \cite{huang2015buffer,han2016mp}. The  average quality of the streamed video is an important QoE metric. 

The key challenge in quantification of stall duration is the choice of scheduling strategy to choose the storage servers for each request, as well as the parallel stream from the chosen servers. For a single video-chunk and single quality videos, the problem is equivalent to minimizing the download latency. This problem is an  open problem, since the optimal  strategy of choosing these $k$ servers (when file is erasure coded with parameters $(n,k)$)  would need a Markov approach similar to that in  \cite{MDS-Queue} which suffers from a  state explosion  problem. Further, the choice of video quality makes the problem challenging since the choice of video quality would depend on the current queue states.  The authors of \cite{Xiang:2014:Sigmetrics:2014,Yu_TON} proposed a probabilistic scheduling method for file scheduling,  where each possibility of $k$ servers is chosen with certain probability that can be optimized. In this paper, we extend this scheduling to a two-stage probabilistic scheduling which chooses $k$ servers and one of the parallel streams from each of these $k$ servers. Further, the choice of video quality is chosen independent of the scheduling and is chosen by a discrete probabilistic distribution. Thus, the proposed scheduling and quality assignment do not account for the current queue state making the approach manageable for analysis.

%The key quality of experience (QoE) metric for video streaming is the duration of  stalls at the clients and the streamed average quality. This paper gives a bound on the stall duration, and uses that to propose an optimized streaming service that optimize the considered QoEs metrics for the clients.

%In this paper, we consider two measures of QoE metrics including mean stall duration and average video quality.   Every viewer almost relates the quality of experiences for watching videos to the stall duration and is thus one of the key focus in the studied streaming algorithms \cite{huang2015buffer,han2016mp}. The second is the average quality of the streamed video. Quantifying service latency for erasure-coded storage is still an open problem \cite{MDS-Queue}, \cite{Jingxian}. 

%This paper takes a step forward and explores the notions for video streaming rather than video download. Thus, finding the exact QoE metrics is an open problem. This paper gives a bound on the stall duration of the streamed video. 

The data chunk transfer time in practical systems follows a shifted exponential distribution \cite{Yu_TON,CS14} which motivates the choice that the service time distribution for each video server is a shifted exponential distribution. Further, the request arrival rates for each video is assumed to be Poisson. The video segments are encoded using an $(n,k)$ erasure code and the coded segments are placed on $n$ different servers. When a video is requested, the segments need to be requested from $k$ out of $n$ servers as well as one of the parallel streams from each of the $k$ servers. Using the two-stage probabilistic scheduling and probabilistic quality assignment, the random variables corresponding to the times for download of different video segments from each server are characterized. By using ordered statistics over the $k$ parallel streams (one from each of the chosen $k$ servers), the random variables corresponding to the playback time of each video segment are then calculated. These are then used to find a bound on the mean stall duration. Moment generating functions of the ordered statistics of different random variables are used in the bound.  We note that the problem of finding latency for file download is very different from the video stall duration for streaming. This is because the stall duration accounts for download time of each video segment rather than only the download time of the last video segment. Further, the download time of segments are correlated since the download of chunks from a server are in sequence and the playback time of a video segment are dependent on the playback time of the last segment and the download time of the current segment. Taking these dependencies into account, this paper characterizes the bound on the mean stall duration.

%In this paper, the probabilistic scheduling proposed in \cite{Xiang:2014:Sigmetrics:2014,Yu_TON} is used to access the $k$ servers, where each possibility of $k$ servers is chosen with certain probability that can be optimized. Using this scheduling policy, 

%We note that for the special case when each video has a single segment, the bounds on mean stall duration and stall duration tail probability reduce to that for file download. Further, the bounds based on the approach in this paper have been shown to outperform the results for mean file download latency in \cite{Xiang:2014:Sigmetrics:2014,Yu_TON}.

%in a recursive manner on the download times and the last play times. 

%and playback time of the segments are found. Using Laplace-Stieltjes transform of these random variables and ordered-statistic analysis, bounds on the  mean stall duration and video stall duration tail probability are derived.

%The proposed framework provides a mathematical crystallization of the engineering artifacts involved and illuminates key system design issues through optimization of QoE. 

A convex combination of mean stall duration and average video quality is optimized over the choice of two-stage probabilistic scheduling, video quality assignment probability, bandwidth allocation among different streams, and the auxiliary variables in the bounds. Changing the convex combination parameter gives a tradeoff between the mean stall duration and the average video quality.  An efficient algorithm is proposed to solve this non-convex problem. The proposed algorithm performs an alternating optimization over the different parameters, where each sub-problem is  shown to have convex constraints and thus can be efficiently solved using iNner cOnVex
Approximation (NOVA)  algorithm proposed in \cite{scutNOVA}. The proposed algorithm is shown to converge to a local optimal. Evaluation results demonstrate significant improvement of QoE metrics as compared to the considered baselines. The key contributions of our paper are summarized as follows. 
\begin{itemize}[leftmargin=0cm,itemindent=.3cm,labelwidth=\itemindent,labelsep=0cm,align=left]
\item This paper proposes a two-stage probabilistic scheduling for the choice of servers and the parallel streams. Further, the video quality is chosen using a discrete probability distribution. 

\item Two-stage probabilistic scheduling and probabilistic quality assignment are used to find the distribution of the  (random) download time of a chunk of each video segment from a parallel stream. Using ordered statistics, the random variable corresponding to the playback time of each video segment is characterized. This is further used to give bounds on the mean stall duration. 

\item  The QoE metrics of mean stall duration and average video quality are used to formulate an optimization problem over the  two-stage probabilistic scheduling access policy, probabilistic quality assignment, the bandwidth allocation weights among the different streams,  and the auxiliary bound parameters which are related to the moment generating function. Efficient iterative solutions are provided for these  optimization problems.
\item The experimental results validate our theoretical analysis and
demonstrate the efficacy of our proposed algorithm. Further, numerical results show that the proposed algorithms
converge within a few iterations. Further, the QoE metrics are shown to have significant improvement as compared to the considered baselines. Even for the minimum stall point, the proposed algorithm gets better quality than the lowest  quality. Further, the tradeoff between stalls and quality can be used by the service provider to effectively find an operating point. 

\end{itemize}

The remainder of this paper is organized as follows. Section \ref{sec:RelWork}
provides related work for this paper. In Section \ref{sec:SysMod}, we describe
the system model used in the paper with a description of
video streaming over cloud storage. Section \ref{sec:DownLoadPlay} derives expressions for the download and play times of the chunks which are used in  Section \ref{sec:mean}  to find an upper bound on the mean stall duration. Section \ref{sec:probForm} formulates the  QoE optimization problem as a weighted combination of the two QoE metrics  and  proposes the iterative algorithmic solution of this problem.   Numerical results are provided in Section \ref{sec:num}. Section \ref{sec:conc} concludes the paper.

\section{Related Work}\label{sec:RelWork}

{\em Latency in Erasure-coded Storage: } To our best knowledge, however, while latency in erasure coded storage systems has been widely studied, quantifying exact latency for erasure-coded storage system in data-center network is an open problem. %Prior works focusing on asymptotic queuing delay behaviors \cite{Bramson:10,Lu:10} are not applicable because redundancy factor in practical data centers typically remains small due to storage cost concerns. %Due to the lack of analytic latency models, most of the literature is focused on reliable distributed storage system design, and latency is only presented as a performance metric when evaluating the proposed erasure coding scheme, e.g., \cite{DPR04,AJX05,HS07,J06, WXHH06}, which demonstrate latency improvement due to erasure coding in different system implementations. Related design can also be found in data access scheduling \cite{SH07,A98,TI10}, access collision avoidance \cite{KL98,ZA02}, and encoding/decoding time optimization \cite{SX,WK} and there is also some work using the LT erasure codes to adjust the system to meet user requirements such as availability, integrity and confidentiality \cite{AG14}. 
Recently, there has been a number of attempts at finding latency bounds for an erasure-coded storage system \cite{MG1:12,Joshi:13,MDS-Queue,Xiang:2014:Sigmetrics:2014,Yu_TON,CS14}. The key scheduling approaches include {\em block-one-scheduling} policy that only allows the request at the head of the buffer to move forward \cite{MG1:12},  fork-join queue \cite{Makowski:89,Joshi:13} to request data from all server and wait for the first $k$ to finish, and the probabilistic scheduling \cite{Xiang:2014:Sigmetrics:2014,Yu_TON} that allows choice of every possible subset of $k$ nodes with certain probability. Mean latency and tail latency have been characterized in \cite{Xiang:2014:Sigmetrics:2014,Yu_TON} and \cite{Jingxian,Tail_TON}, respectively, for a system with multiple files using probabilistic scheduling. The probabilistic scheduling has also been shown to be optimal for tail latency index when the file sizes are heavy-tailed \cite{TailIndex}.  This paper considers video streaming rather than file downloading. The metrics for video streaming does not only account for the end of the download of the video but also of the download of each of the segment. Thus, the analysis for the content download cannot be extended to the video streaming directly and the analysis approach in this paper is very different from the prior works in the area.

{\em Video Streaming over Cloud: } Servicing Video on Demand and Live TV Content from cloud servers have been studied widely \cite{lee2013vod,huang2011cloudstream,he2014cost,chang2016novel,oza2016implementation}. 
% cite Aggarwal:multimedia:2013 here for final version
The placement of content and resource optimization over the cloud servers have been considered. To the best of our knowledge, reliability of content over the cloud servers have not been considered for video streaming applications. In the presence of erasure-coding, there are novel challenges to characterize and optimize the QoE metrics at the end user. Adaptive streaming algorithms have also been considered for video streaming \cite{chen2012amvsc,wang2013ames}, which are beyond the scope of this paper and are left for future work. 

Recently, the authors of \cite{Abubakr_TON} considered video-streaming over cloud. However, the videos were a single quality and the quality optimization was not accounted. Further,  \cite{Abubakr_TON}  considered single stream between each storage server and edge node and thus two-stage probabilistic scheduling was not needed. Thus, the analysis and the problem formulation in this work is different from that in \cite{Abubakr_TON}.

\section{System Model}\label{sec:SysMod}

\begin{figure}
	\includegraphics[trim=1.5in 2.5in 1.8in 0.3in, clip, width=0.40\textwidth]{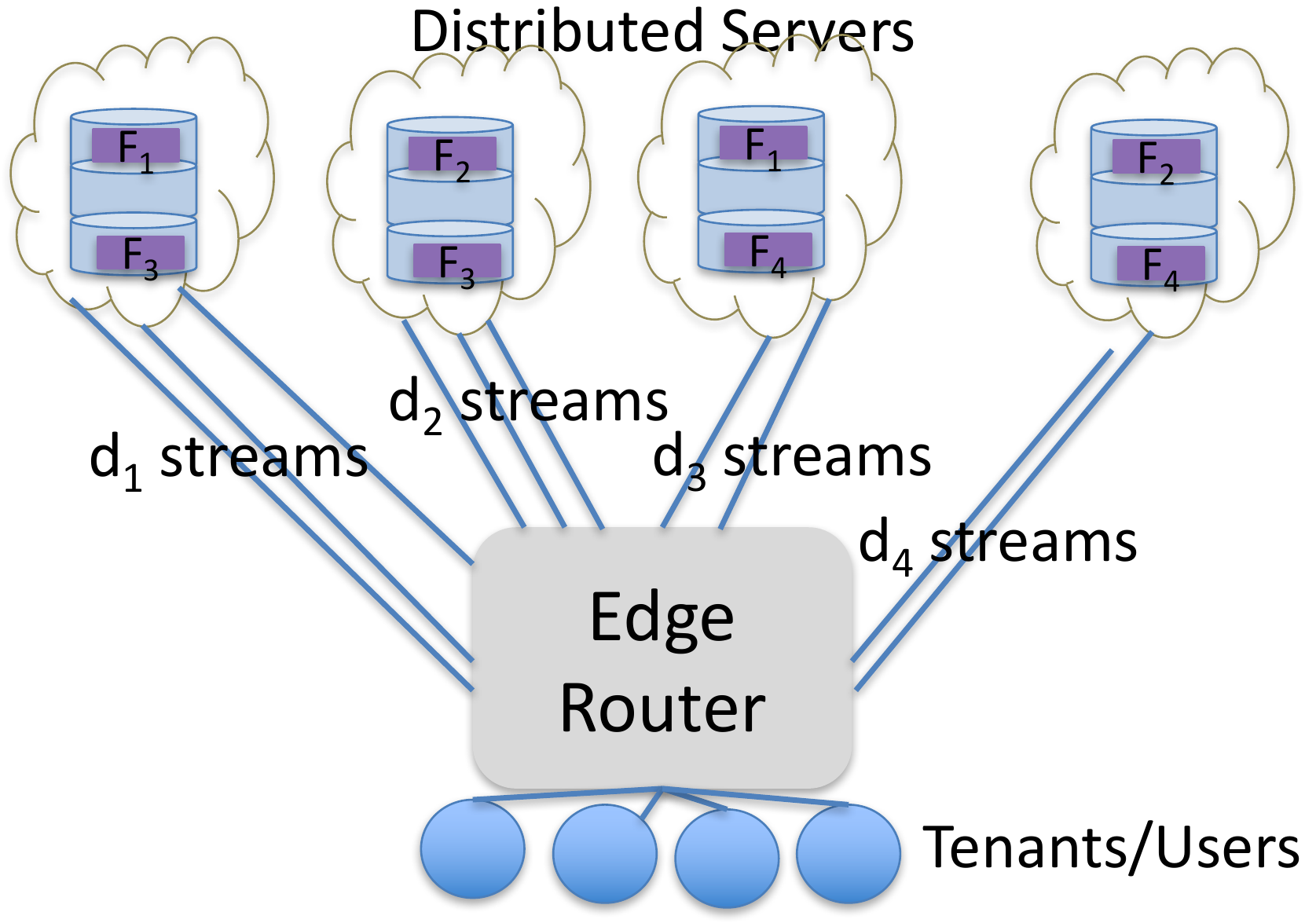}
	\vspace{-.1in}
	\caption{ An Illustration of a distributed storage system
		equipped with $m=4$ nodes. Storage server $j$ has $d_j$ streams to the edge router. \label{fig:sys}}
	\vspace{-.15in}
\end{figure}

We consider a distributed storage system consisting of $m$ heterogeneous servers
(also called storage nodes), denoted by $\mathcal{M}=1,2,\cdots,m$. Each server $j$ can be split into $d_{j}$ virtual outgoing parallel streams (queues) to the edge router, where the server bandwidth is split among all $d_j$ parallel streams (PSs).  This is depicted in Fig. \ref{fig:sys}.  The reason of having $d_j$ PSs is to serve $d_j$ video files simultaneously from a server thus helping one file not to have files wait for the previous long video files. This is a key difference for video streaming as compared to file download since the deadline for the later video chunks are late thus motivating prioritizing earlier chunks. This parallelization helps download  multiple files in parallel which also delays the finishing of download of the last chunks of multiple requests. Multiple users are connected to edge-router, where we assume that the connection between user and edge router is infinite and thus only consider the links from the server to the edge router. Thus, we can consider edge router as an aggregation of multiple users. Let $\left\{ w_{j,\nu_{j}},\forall j= 1, \cdots, m,   \nu_j= 1, \cdots d_j \right\} $ be a set of $d_j$ non-negative
weights representing the split of bandwidth at server $j$ on the $d_j$ PSs. The weights satisfy  $\sum_{\nu_j=1}^{d_j}w_{j,\nu_j}\leq1  \forall j$. 
The sum of weights at all PSs can be smaller than $1$, representing that the bandwidth may not be completely utilized. By optimizing $w_{j,\nu_j}$, the server bandwidth can be efficiently split among different PSs. Optimizing these weights help avoid bandwidth under-utilization and congestion, for example,  assigning larger bandwidth to heavy workload PSs can help reduce mean stall duration.

% because
%the total bandwidth at server $j$, i.e., $\alpha_j^{(\ell)}$, may not be fully utilized as one or more VQs might be unused.

%We consider a distributed storage system consisting of $m$ heterogeneous servers (also called storage nodes), denoted by $\mathcal{M}=1,2,...,m$. Each server $j$ can handle $\nu_{j}$, where $\nu_{j}=\left\{ 1,\ldots,d_{j}\right\} $, jobs in parallel such that the server bandwidth is split among all $d_j$ virtual machines of the $j$ server. 

\begin{figure}
	\includegraphics[trim=0in 0in 0in 0in, clip, scale=0.30]{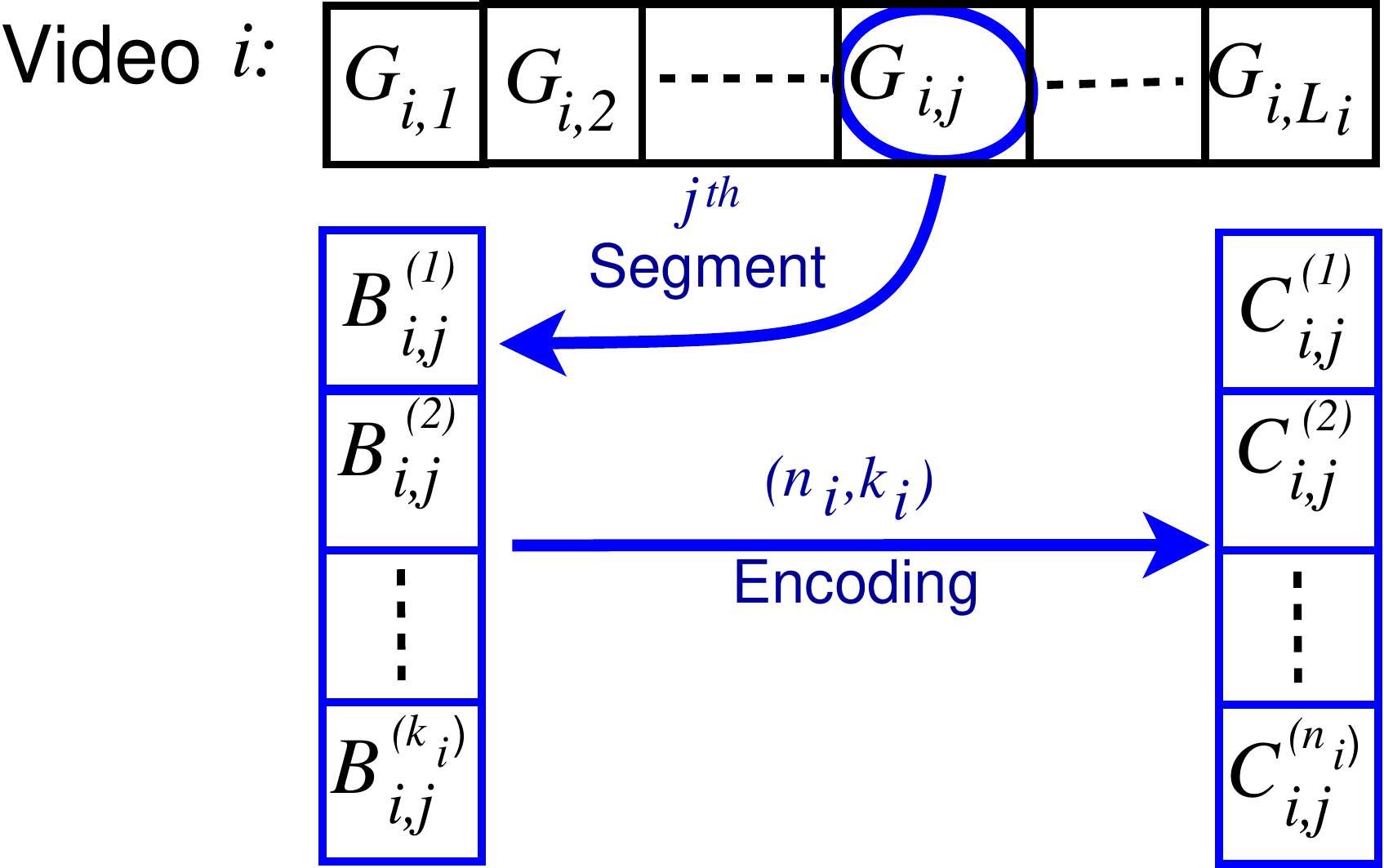}
	\vspace{-.1in}
	\caption{A schematic illustrates video fragmentation and erasure-coding processes. Video $i$ is composed of $L_{i}$ segments. Each segment is partitioned into $k_{i}$ chunks and then encoded using an $(n_{i},k_{i})$ MDS code. The quality index is omitted in the figure for simplicity. \label{fig:videoEncoding}}
	\vspace{-.2in}
\end{figure}

Each video file $i$, where $i=1,2,\cdots, r,$ is divided into $L_{i}$ equal segments, each of length $\tau$ seconds.  We assume that each video file is encoded to different qualities, {\em i.e.}, $\ell\in \{1,2, \cdots,V\}$, where $V$ are the number of possible choices for the quality level. The $L_{i}$  segments of video file $i$ at quality $\ell$ are denoted as $G_{i,\ell,1}, \cdots, G_{i,\ell,L_i}$. Then, each segment $G_{i,\ell,u}$ for $u\in\left\{ 1,2,\ldots,L_{i}\right\} $ and $\ell\in \{1,2, \cdots,V\}$ is  partitioned  into $k_i$ fixed-size chunks  and then
 encoded  using an $(n_i, k_i)$ Maximum Distance Separable (MDS) erasure code to generate $n_i$  distinct chunks for each segment $G_{i,\ell,u}$. These coded chunks are denoted as $C_{i,\ell,u}^{(1)}, \cdots, C_{i,\ell,u}^{(n_i)}$. The encoding setup is illustrated in Figure \ref{fig:videoEncoding}. The encoded chunks for all quality levels are stored on the disks of $n_i$ distinct storage nodes. The storage nodes chosen for quality level $\ell$ are represented by a set $\mathcal{S}_{i}^{(\ell)}$, such that 
 $\mathcal{S}_{i}^{(\ell)}\subseteq\mathcal{M}$ and $n_{i}=\left|\mathcal{S}_{i}^{(\ell)}\right|$. Each server $z\in \mathcal{S}_{i}^{(\ell)}$ stores all the chunks $C_{i,\ell, u}^{(g)}$ for all $u$ and for some $g$. In other words, $n_i$ servers store the entire content, where a server stores coded chunk $g$ for all the video-chunks for some $g$ or does not store any chunk.  We will use a probabilistic quality assignment strategy, where  a chunk of quality $\ell$  of size $a_{\ell}$  is requested with probability $b_{i,\ell}$ for all $\ell\in \{1,2, \cdots,V\}$. We further  assume  all the chunks of the video are fetched at the same quality level.  Note that $k_i=1$ indicates that the video file $i$ is replicated $n_i$ times.

In order to serve the incoming request at the edge router, the video can be reconstructed from the video chunks from any subset  of $k_{i}$-out-of-$n_{i}$ servers. Further, we need to assign one of the $d_j$ PSs for each server $j$ that is selected. We assume that files at each PS are served in order of the request in a first-in-first-out (FIFO) policy. Further, the different video chunks in a video are processed in order. In order to select the different PSs for video $i$ and quality $\ell$, the request goes to a set $\mathcal{A}_{i}^{(\ell)} = \{(j,\nu_j): j\in \mathcal{S}_{i}^{(\ell)}, \nu_j \in \{1, \cdots, d_j\}\}$, with $\left|\mathcal{A}_{i}^{(\ell)}\right|=k_i$ and  for every $(j,\nu_j)$ and $(k,\nu_k)$ in $\mathcal{A}_{i}^{(\ell)}$, $j\ne k$. Here, the choice of $j$ represents the server to choose and $\nu_j$ represents the PS selected. From each choice $(j,\nu_j) \in \mathcal{A}_{i}^{(\ell)}$, all chunks $C_{i,\ell,u}^{(g)}$ for all $u$ and the value of $g$ corresponding to that placed on server $j$ are requested from PS $\nu_j$. The choice of optimal scheduling strategy, or set $\mathcal{A}_{i}^{(\ell)}$ is an open problem. In this paper, we extend the probabilistic scheduling  proposed in \cite{Xiang:2014:Sigmetrics:2014,Yu_TON} to two-stage probabilistic scheduling. The two-stage probabilistic scheduling chooses every possible subset of $k_i$-out-of-$n_i$ nodes with certain probability, and for every chosen node $j$, chooses $1$-out-of-$d_j$ PSs with certain probability.  Let $\pi_{i,j,\nu_{j}}^{(\ell)}$ is the probability of requesting file $i$ from the PS $\nu_{j}$ that belongs to server $j$ for quality level $\ell$. Thus, $\pi_{i,j,\nu_{j}}^{(\ell)}$ is given by
\begin{equation}
\pi_{i,j,\nu_{j}}^{(\ell)}=q_{i,j}^{(\ell)}p_{j,\nu_j}^{(\ell)},\label{eq:pi_i_j_nu}
\end{equation}
where $q_{i,j}^{(\ell)}$ is the probability of choosing server $j$ and $p_{j,\nu}^{(\ell)}$ is the probability of choosing  PS $\nu_j$ at server $j$. Following  \cite{Xiang:2014:Sigmetrics:2014,Yu_TON}, it can be seen that the two-stage probabilistic scheduling gives feasible probabilities for choosing $k_i$-out-of $n_i$ nodes and one-out-of-$d_j$ PSs  if and only if there exists conditional probabilities $q_{i,j}^{(\ell)}\in\left[0,1\right]$ and $p_{j,\nu_j}^{(\ell)}\in\left[0,1\right]$ satisfying
\begin{equation}
\sum_{j=1}^{m}q_{i,j}^{(\ell)}=k_{i}\,\,\,\,\forall i\,\,\,\,\,\,\,\,\,\mbox{and}\,\,\,\,\,\,q_{i,j}^{(\ell)}=0\,\,\,\,\,\mbox{if\,\,\,\ensuremath{j\notin \mathcal{S}_{i}^{(\ell)}}},
\end{equation}
and
\begin{equation}
\sum_{\nu_j=1}^{d_j}p_{j,\nu_j}^{(\ell)}=1\,\,\,\,\forall j.
\end{equation}

We now describe a queuing model of the distributed storage system.
We assume that the arrival of  requests at the edge router for each video $i$ %of quality $\ell$%
form
an independent Poisson process with a known rate $\lambda_{i}$.  Using the two stage probabilistic scheduling and the quality assignment probability distribution, the arrival of file requests at PS $\nu_j$ at node $j$ forms a Poisson Process
with rate $\varLambda_{j,\nu_j}=\sum_{i,\ell}\lambda_{i}\pi_{i,j,\nu_j}^{(\ell)} b_{i, \ell}$ which is
the superposition of $r d_j$ Poisson processes each with rate $\lambda_{i}\pi_{i,j,\nu_j}^{(\ell)}b_{i, \ell}$.  We assume that the chunk service time for each coded chunk $C_{i,\ell,u}^{(g)}$ at PS $\nu_j$ of server $j$,  $X_{j,\nu_j}^{(\ell)}$,  follows a shifted exponential distribution as has been demonstrated in realistic systems \cite{Yu_TON,CS14} and is given by
the probability distribution function $f_{j,\nu_j}^{(\ell)}(x)$, which is
\begin{equation}
f_{j,\nu_j}^{(\ell)}(x)=\begin{cases}
\begin{array}{cc}
\alpha_{j,\nu_j}^{(\ell)}e^{-\alpha_{j,\nu_j}^{(\ell)}\left(x-\beta_{j,\nu_j}^{(\ell)}\right)}\,, & \,\,\,\,\,x\geq\beta_{j,\nu_j}^{(\ell)}\\
0\,, & \,\,\,\,\,\,x<\beta_{j,\nu_j}^{(\ell)}
\end{array}\end{cases}.
\end{equation}

We note that exponential distribution is a special case with $\beta_{j,\nu_j}^{(\ell)}=0$. Let $M_{j,\nu_j}^{(\ell)}(t)=\mathbb{E}\left[e^{tX_{j,\nu_j}^{(\ell)}}\right]$ be the moment generating function of $X_{j,\nu_j}^{(\ell)}$ whose quality is ${\ell}$.  Then, $M_{j,\nu_j}^{(\ell)}(t)$ is given as

\begin{equation}
M_{j,\nu_j}^{(\ell)}(t)=\frac{\alpha_{j,\nu_j}^{(\ell)}}{\alpha_{j,\nu_j}^{(\ell)}-t}\,e^{\beta_{j,\nu_j}^{(\ell)}t}\,\,\,\,\,\,\,\,\,\,t<\alpha_{j,\nu_j}^{(\ell)} 
\label{M_j_t_1}
\end{equation}

Note that the value of  $\beta_{j,\nu_j}^{(\ell)}$ increases in proportion to the
chunk size, and the value of $\alpha_{j,\nu_j}^{(\ell)}$ decreases in proportion to the chunk size in the shifted-exponential service time distribution. Further, the rate $\alpha_{j,\nu_j}^{(\ell)}$ is proportional to the assigned bandwidth $w_{j,\nu_j}$. More formally, the parameters $\alpha_{j,\nu_j}^{(\ell)}$  and $\beta_{j,\nu_j}^{(\ell)}$ are given as
\begin{equation}
\alpha_{j,\nu_j}^{(\ell)} = \alpha_j w_{j,\nu_j}/a_\ell, \   \beta_{j,\nu_j}^{(\ell)} = \beta_j a_\ell, \label{eq:alphabeta}
\end{equation}
where $\alpha_j$ and $\beta_j$ are constant service time parameters when $a_\ell=1$ and the entire bandwidth is allocated to one PS. Since $ \beta_{j,\nu_j}^{(\ell)}$ mainly represents the read time and other processing times, we assume that all PSs have the same value of  $\beta_{j,\nu_j}^{(\ell)}$.

%To the best of our knowledge, this is the first work jointly optimizing stall duration
%of video services in an erasure-coded system along with the quality of the streamed video files. 

%To accommodate network traffic received by each server $j$, we manage all
%incoming requests generated by applications at sever $j$ in a (virtual)
%local queues (VQs), denoted by $\boldsymbol{Q}_{j}^{(\nu_j)}$, where $\nu_j\in\{1,2,\ldots,d_j\}$.
%Therefore, each server $j$ contains $d_j$ dependent queues (corresponding to the $d_j$ PSs) and the
%server bandwidth $\alpha_{j}^{(\ell)}$ available at each server $j$ has to be
%apportioned among the different PSs in a coordinated and optimized
%manner.

We note that the arrival rates are given in terms of the video files, and the service rate above is provided in terms of the coded chunks at each server. The client plays the video segment after all the $k_i$ chunks for the segment have been downloaded and the previous segment has been played. We also assume that there is a start-up delay of $d_{s}$ (in seconds) for the video which is the duration in which the content can be buffered but not played.  This paper will characterize the mean stall duration using two-stage probabilistic scheduling and probabilistic quality assignment.

\section{Download and Play Times of the Chunks} \label{sec:DownLoadPlay}

In order to understand the stall duration, we need to see the download time of different coded chunks and the play time of the different segments of the video. 

\subsection{Download Times of the Chunks from each Server}
In this subsection, we will quantify the download time of chunk for video file $i$ from server $j$ which has  chunks $C_{i,\ell,u}^{(g)}$ for all $u = 1, \cdots L_i$. The download of $C_{i,\ell,u}^{(g)}$ consists of two components - the waiting time of the video files in the queue of the PS before file $i$ request and the service time of all chunks of video file $i$ up to the $g^{\text{th}}$ chunk. Let  $\ensuremath{W_{j,\nu_j}}$ be the random variable corresponding to the waiting time of all the video files in queue of PS $\nu_j$ at server $j$ before file $i$ request and $Y_{j,\nu_j}^{(g,\ell)}$ be the (random) service time of coded chunk $g$ for file $i$ with quality $\ell$ from PS $\nu_j$ at server $j$. Then, the (random) download time for coded chunk $u\in \{1, \cdots, L_i\}$ for file $i$ at PS $\nu_j$ at server $j\in \mathcal{A}_{i}^{(\ell)}$, $D_{i,j}^{(u,\ell)}$, is given as 
\begin{equation}
D_{i,j,\nu_j}^{(u,\ell)} = W_{j,\nu_j} + \sum_{v=1}^u Y_{j,\nu_j}^{(v,\ell)}. \label{dije}
\end{equation}

%We consider download of $q^{\text{th}}$  chunk $C_{i,\ell,u}^{(g)}$. As seen in Figure \ref{fig:plcOnServ}, the download of $C_{i,u}^{(g)}$ consists of two components - the waiting time of the video files in the VQs before file $i$ request and the service time of all chunks of video file $i$ up to the $g^{\text{th}}$ chunk. 

We will now find the distribution of $W_{j,\nu_j}$. We note that this is the waiting time for the video files whose arrival rate is given as $\varLambda_{j,\nu_j}=\sum_{i,\ell}\lambda_{i}b_{i,\ell}\pi_{i,j,\nu_j}^{(\ell)}$. In order to find the waiting time, we would need to find the service time statistics of the video files. Note that $f_{j,\nu_j}^{(\ell)}(x)$ gives the service time distribution of only a chunk and not of the video files. 

Video file $i$ of quality $\ell$ consists of $L_{i}$ coded chunks at PS $\nu_j$ at server $j$ ($j\in \mathcal{S}_{i}^{(\ell)}$). The total service time for video file $i$ with quality $\ell$ at PS $\nu_j$ at server $j$ if requested from server $j$, $ST_{i,j,\nu_j}^{(\ell)} $, is given as  
\begin{equation}
ST_{i,j,\nu_j}^{(\ell)} = \sum_{v=1}^{L_i} Y_{j,\nu_j}^{(v,\ell)}.
\end{equation}

The service time of the video files is given as 
\begin{equation}
R_{j,\nu_j} = \begin{cases}
ST_{i,j,\nu_j}^{(\ell)}  \quad \text{ with probability } \frac{\pi_{i,j,\nu_j}^{(\ell)}\lambda_{i}b_{i,\ell}}{\Lambda_{j,\nu_j}} \quad \forall i,\ell, 
\end{cases}
\end{equation}
since the service time is $ST_{i,j,\nu_j}^{(\ell)} $ when file $i$ is requested at quality $\ell$ from PS $\nu_j$ from server $j$. Let $\overline{R}_{j,\nu_j}(s) = {\mathbb E}[e^{-sR_{j,\nu_j} }]$ be the Laplace-Stieltjes Transform of $R_{j,\nu_j}$. 

\begin{lemma}\label{LJ:servTimeofFile}
	The Laplace-Stieltjes Transform of $R_{j,\nu_j}$, $\overline{R}_{j,\nu_j}(s)=\mathbb{E}\left[e^{-s\overline{R}_{j,\nu_j}}\right]$ is given as
	\begin{equation}
\overline{R}_{j,\nu_j}(s)  = \sum_{i=1}^r \sum_{\ell=1}^{V} \frac{\pi_{i,j,\nu_j}^{(\ell)}\lambda_ib_{i,\ell}}{\Lambda_{j,\nu_j}}	\left(\frac{\alpha_{j,\nu_j}^{(\ell)}e^{-\beta_{j,\nu_j}^{(\ell)}s}}{\alpha_{j,\nu_j}^{(\ell)}+s}\right)^{L_{i}}\label{eq:servTimeofFile}
	\end{equation}
	\end{lemma}

\begin{proof}

\begin{align}
\overline{R}_{j,\nu_j}(s) & =\sum_{i=1}^r \sum_{\ell=1}^{V} \frac{\pi_{i,j,\nu_j}^{(\ell)}\lambda_{i}b_{i,\ell}}{\Lambda_{j,\nu_j}}\mathbb{E}\left[e^{-s\left(ST_{i,j,\nu_j}^{(\ell)}\right)}\right]\nonumber \\
 & \overset{}{=}\sum_{i=1}^r \sum_{\ell=1}^{V}  \frac{\pi_{i,j,\nu_j}^{(\ell)}\lambda_{i}b_{i,\ell}}{\Lambda_{j,\nu_j}}\mathbb{E}\left[e^{-s\left(\sum_{\nu=1}^{L_{i}}Y_{j,\nu_j}^{(\nu,\ell)}\right)}\right]\nonumber \\
 & =\sum_{i=1}^r \sum_{\ell=1}^{V} \frac{\pi_{i,j,\nu_j}^{(\ell)}\lambda_{i}b_{i,\ell}}{\Lambda_{j,\nu_j}}\left(\mathbb{E}\left[e^{-s\left(Y_{j,\nu_j}^{(1,\ell)}\right)}\right]\right)^{L_{i}}\nonumber \\
 & =\sum_{i=1}^r \sum_{\ell=1}^{V} \frac{\pi_{i,j,\nu_j}^{(\ell)}\lambda_{i}b_{i,\ell}}{\Lambda_{j,\nu_j}}\left(\frac{\alpha_{j,\nu_j}^{(\ell)}e^{-\beta_{j,\nu_j}^{(\ell)}s}}{\alpha_{j,\nu_j}^{(\ell)}+s}\right)^{L_{i}}
\end{align}

\end{proof}

\begin{corollary}
	The moment generating function for the service time of video files when requested from server $j$ and PS $\nu_j$, $B_{j,\nu_j}(t)$, is given as
	\begin{equation}
B_{j,\nu_j}(t)  = \sum_{i=1}^r \sum_{\ell=1}^{V} \frac{\pi_{i,j,\nu_j}^{(\ell)}\lambda_i b_{i,\ell}}{\Lambda_{j,\nu_j}}	\left(\frac{\alpha_{j,\nu_j}^{(\ell)}e^{\beta_{j,\nu_j}^{(\ell)}t}}{\alpha_{j,\nu_j}^{(\ell)}-t}\right)^{L_{i}}\label{eq:servTimeofFileB_j_i}
	\end{equation}
	for any $t>0$, and $t< \alpha_{j,\nu_j}$.
	\end{corollary}	
\begin{proof}
This corollary follows from (\ref{eq:servTimeofFile})
 by setting $t=-s$.
\end{proof}

The server utilization for the video files at PS $\nu_j$ of server $j$ is given as $\rho_{j,\nu_j}=\varLambda_{j,\nu_j}\mathbb{E}\left[R_{j,\nu_j}\right]$. Since $\mathbb{E}\left[R_{j,\nu_j}\right] = B_{j,\nu_j}'(0)$, using Lemma \ref{LJ:servTimeofFile}, we have

%for erasure-coded storage
%systems with known service time distribution $X_{j}$.  

% Preview source code from paragraph 14 to 24

%For a file $i$, which is composed of $L_{i}$ segments, the total service time is the sum of
%service times of all segments, hence, the Laplace transform
%of the service time of a file $i$ is given by 

%\begin{equation}
%B_{j}^{(i)}\left(s\right)=\left(\frac{\alpha_{j}e^{-\beta_{j}s}}{\alpha_{j}+s}\right)^{L_{i}}\label{eq:servTimeofFile}
%\end{equation}
%Proof: See Appendix.

%Note that, unlike earlier works, e.g., \citep{xxx,xxx}, the video file requests in the same server have different service time with larger service times go to longer video files, as expected. 

%Using (\ref{eq:servTimeofFile}), the expected service time of a file $i$ is given by

%\begin{align}
%\mathbb{E}\left[B_{j}^{(i)}\left(t_{i}\right)\right] & =\frac{\sum_{i}\pi_{ij}\lambda_{i}}{\Lambda_{j}}B_{j}^{(i)}\left(t_{i}\right)\nonumber \\
% & =\frac{\sum_{i}\pi_{ij}\lambda_{i}}{\Lambda_{j}}\left(\frac{\alpha_{j}e^{\beta_{j}t_{i}}}{\alpha_{j}-t_{i}}\right)^{L_{i}}\nonumber \\
% & =\overline{B_{j}^{(i)}}\left(t_{i}\right)\label{eq:ExpfileServTime}
%\end{align}
%
%Prrof. See Appendix.

%Since $\rho_{j}=\varLambda_{j}\mathbb{E}\left[X_{j}\right]$, we can get the request intensity at node $j$ as follows
\begin{equation}
\rho_{j,\nu_j}=\sum_{i=1}^{r} \sum_{\ell=1}^{V} \pi_{i,j,\nu_j}^{(\ell)} \lambda_{i} b_{i,\ell}L_{i}\left(\beta_{j,\nu_j}^{(\ell)}+\frac{1}{\alpha_{j,\nu_j}^{(\ell)}}\right)\label{eq:rho_j_1}.
\end{equation}

%Further, let $\ensuremath{W_{j}}$ be the random waiting time the chunk request spends in node $j$. Under probabilistic scheduling, the service time (sojourn time) is determined by the maximum chunk service time at a randomly selected set $\mathcal{A}_{i}$ of storage nodes. 

Having characterized the service time distribution of the video files via a Laplace-Stieltjes Transform $\overline{R}_{j,\nu_j}(s) $, the Laplace-Stieltjes Transform of the waiting time $W_{j,\nu_j}$ can be characterized using Pollaczek-Khinchine formula for M/G/1 queues \cite{zwart2000sojourn}, since the request pattern is Poisson and the service time is general distributed. Thus, the Laplace-Stieltjes Transform of the waiting time $W_{j,\nu_j}$ is given as 

%Then, the Laplace Stieltjes Transform of $\ensuremath{W_{ت}}$ is given, using Pollaczek-Khinchine formula, as follows
\begin{equation}
\mathbb{E}\left[e^{-sW_{j,\nu_j}}\right]=\frac{\left(1-\rho_{j,\nu_j}\right)s\overline{R}_{j,\nu_j}(s)}{s-\Lambda_{j,\nu_j}\left(1-\overline{R}_{j,\nu_j}(s)\right)}\label{eq:E_W_j_laplace}
\end{equation}

%Also, let $\ensuremath{D_{i}^{(q)}}$ denotes
%the (max) download time of the video segment $q$ and also let $Y_{j}^{(q)}$ denotes the service time of $q$ segment at node $j$. Hence, $\ensuremath{D_{i}^{(q)}}$ can be expressed as follows 
%\begin{equation}
%D_{i}^{(q)}=\underset{j\in\mathcal{S}_{i}}{\text{max}}\left(W_{j}+Y_{j}^{(q)}\right).\label{eq:D_i_qi}
%\end{equation}
%where $X_j^{(q)} = X_j^{(q-1)}+Y_j^{(q)}$, and $Y_{j}^{(q)}$ has distribution of $f_{j}\left(x\right)$.

By characterizing the Laplace-Stieltjes Transform of the waiting time $W_{j,\nu_j}$ and knowing the distribution of  $Y_{j,\nu_j}^{(v,\ell)}$, the Laplace-Stieltjes Transform of the download time $D_{i,j,\nu_j}^{(u,\ell)}$ is given as

\begin{equation}
{\mathbb E}[e^{-sD_{i,j,\nu_j}^{(u,\ell)}}] = \frac{\left(1-\rho_{j,\nu_j}\right)s\overline{R}_{j,\nu_j}(s)}{s-\Lambda_{j,\nu_j}\left(1-\overline{R}_{j,\nu_j}(s)\right)}\left( \frac{\alpha_{j,\nu_j}^{(\ell)}\,e^{-\beta_{j,\nu_j}^{(\ell)}s}}{\alpha_{j,\nu_j}^{(\ell)}+s}\right)^u.\label{LapOfE_D_ij}
\end{equation}

We note that the expression above holds only in the range of $s$ when  $s-\Lambda_{j,\nu_j}\left(1-\overline{R}_{j,\nu_j}(s)\right)>0$ and $\alpha_{j,\nu_j}^{(\ell)}+s>0$. Further, the server utilization $\rho_{j,\nu_j} $ must be less than $1$. The overall download time of all the chunks for the segment $G_{i,u,\ell}$ at the client,  $D_{i}^{(u,\ell)}$, is given by 

\begin{equation}
D_{i}^{(u,\ell)} = \max_{(j,\nu_j)\in \mathcal{A}_{i}}  D_{i,j,\nu_j}^{(u,\ell)}. \label{deq}
\end{equation}

\subsection{Play Time of Each Video Segment}
 Let $T_{i}^{\left(u,\ell\right)}$ be the  time at which the segment $G_{i,\ell,u}$ is played (started) at the client. The startup delay of the video is $d_s$. Then, the first segment can be played at the maximum of the time the first segment can be downloaded and the startup delay. Thus, 

%Now, we explain our approach in characterizing the mean stall duration. First, we define $d_{s}$ as the start-up delay in Seconds (i.e., start-up buffering delay time). Next, note that the stall time could be either zero (the segment to be played is received on time or earlier) or positive
%number (the segment is received later than its expected played time). For example, a stall could happen if the first segment is received later than $d_{s}$, i.e., $D_{i}^{(1)}$ is larger than $d_{s}$. Mathematically, the play time of the first segment is given by 
\begin{eqnarray}
T_{i}^{(1,\ell)} & = & \mbox{max }\left(d_{s},\,D_{i}^{(1,\ell)}\right).
\label{eq:T_i_qi_1}
\end{eqnarray}

For $1<u\le L_i$, the  play time of  segment $u$ of file $i$ is given by the maximum of the time it takes to download the segment and the time at which the previous segment is played plus the time to play a segment ($\tau$ seconds). Thus, the play time of segment $u$ of file $i$, $T_{i}^{(u,\ell)}$  can be expressed as 
\begin{eqnarray}
T_{i}^{(u,\ell)} & = & \mbox{max }\left(T_{i}^{(u-1,\ell)}+\tau,\,D_{i}^{(u,\ell)}\right).
\label{eq:T_i_qi}
\end{eqnarray}

%Recall that, $L_{i}$ is the last segment of video $i$, i.e., number of segments of video $i$. Now, we estimate an upper bound on the total expected stall
%time $\mathbb{E}\left[T_{S}^{(i)}\right]$. To achieve this goal, we need to evaluate the expected play time of the last segment, i.e., $\mathbb{E}\left[T_{i}^{(L_{i})}\right]$. Notice that $\mathbb{E}\left[T_{i}^{(L_{i})}\right]$ is a maximum of $(L_{i}+1)$ random variables. From (\ref{eq:T_i_qi}), the play time of the last segment is given by

%\begin{eqnarray}
%T_{i}^{(L_{i})} & = & \mbox{max }\left(T_{i}^{(L_{i}-1)}+\tau,\,D_{i}^{(L_{i})}\right)
%\end{eqnarray}

Equation \eqref{eq:T_i_qi} gives a recursive equation, which can yield

\begin{eqnarray}
T_{i}^{(L_{i},\ell)} & = & \mbox{max }\left(T_{i}^{(L_{i}-1,\ell)}+\tau,\,D_{i}^{(L_{i},\ell)}\right)\nonumber\\
 & = & \mbox{max }\left(T_{i}^{(L_{i}-2,\ell)}+2\tau,\,D_{i}^{(L_{i}-1,\ell)}+\tau,\,D_{i}^{(L_{i},\ell)}\right)\nonumber\\
\!\!& = & \!\!\!\!\mbox{max}\!\left(\!\mathcal{F}_{j,1,\nu_j,\ell},\max_{z=2}^{L_{i}+1}D_{i}^{(z-1,\ell)}+(\!L_{i}\!-z\!+1)\!\tau\right)\ \label{eq:T_{i}{}_{q}i2} \label{eq:T_i_qi2}
\end{eqnarray}
where 
\begin{eqnarray}
\mathcal{F}_{j,z,\nu_j,\ell}=\begin{cases}
d_{s}+\left(L_{i}-1\right)\tau &,\,\, z=1\\
\\
D_{i,j,\nu_j}^{(z-1,\ell)} + (L_i-z+1)\tau \ &,\,\, 2\leq z\leq (L_{i}+1)
\end{cases}.\label{eq:pjz}
\end{eqnarray}

Since $D_{i}^{(u,\ell)} = \max_{(j,\nu_j)\in \mathcal{A}_{i}^{(\ell)}}^{(\ell)}  D_{i,j,\nu_j}^{(u,\ell)}$ from \eqref{deq}, $T_{i}^{(L_{i},\ell)}$ can  be written as

\begin{eqnarray}
T_{i}^{(L_{i},\ell)}= \max_{z=1}^{L_i+1}\max_{(j,\nu_j)\in \mathcal{A}_i}\left(\mathcal{F}_{j,z,\nu_j,\ell}\right)\label{eq:T_i_L_i}. \label{eq:pjstart}
\end{eqnarray}

%$\mathcal{F}_{j,1,\nu_j,\ell}= d_{s}+(L_{i}-1)\tau$. where
%Note that in (\ref{eq:T_i_L_i}), we do first max over $j$ (random dispatch decisions) and then over $z$ (the term index).

We next give the moment generating function of $\mathcal{F}_{j,z,\nu_j,\ell}$ that will be used in the calculations of the mean stall duration in the next section.% Hence, we define the following lemmas. 
%\begin{lemma}
%	The moment generating function %for $\ensuremath{f_{X}{}_{j}(x)}$, is given as
%\begin{equation}
%M_{j}(t)=\frac{\alpha_{j}%e^{\beta_{j}t}}{\alpha_{j}-t}\,\,\,\,\text{for \,\, \ensuremath{t<\alpha_{j}}}.\label{eq:Mgf_t}
%\end{equation}
%	\end{lemma}
		
\begin{lemma}
	The moment generating function for $\mathcal{F}_{j,z,\nu_j,\ell}$, is given as
	\begin{equation}
\mathbb{E}\left[e^{t\mathcal{F}_{j,z,\nu_j,\ell}}\right]=\begin{cases}
e^{t\left(d_{s}+\left(L_{i}-1\right)\tau\right)} & ,\,z=1\\
e^{t\left(L_{i}+1-z\right)\tau}Z_{D_{i,j,\nu_j}^{(z-1,\ell)}}\left(t\right) & ,2\leq z\leq L_{i}+1
\end{cases}
\label{eq:momntPjz}
	\end{equation}
where 

\begin{equation}
Z_{D_{i,j,\nu_j}^{(u,\ell)}}\left(t\right) = {\mathbb E}[e^{tD_{i,j,\nu_j}^{(u,\ell)}}]  = \frac{\left(1-\rho_{j,\nu_j}\right)t_{i}B_{j,\nu_j}(t_{i})\left(M_{j,\nu_j}^{(\ell)}(t_{i})\right)^{u}}{t_{i}-\Lambda_{j,\nu_j}\left(B_{j,\nu_j}(t_{i})-1\right)}\label{eq:M_D_ij} 
\end{equation}
	\end{lemma}

\begin{proof}
This follows by substituting
$t=-s$ in (\ref{LapOfE_D_ij}) and $B_{j,\nu_j}(t_{i})$ is given by  (\ref{eq:servTimeofFileB_j_i}) and $M_{j,\nu_j}^{(\ell)}(t_i)$ is given by (\ref{M_j_t_1}). This expression holds when $t_{i}-\Lambda_{j,\nu_j}\left(B_{j,\nu_j}(t_{i})-1\right)>0$ and $t_i<0  \,\forall j,\nu_j$,  since the moment generating function does not exist if the above do not hold.
\end{proof}

Ideally, the last segment should have started played by time $d_s + (L_i-1) \tau$. The difference between $T_{i}^{(L_i,\ell)}$ and $d_s + (L_i-1) \tau$ gives the stall duration. We note that $T_{i}^{(L_i,\ell)}$ is not the download time of the last segment, but the play time of the last segment and accounts for the download of all the $L_i$ segments. This is a key difference as compared to the file download since the download time of each segment of the video has to be accounted for computing stall duration.   Thus, the stall duration for the request of video file $i$ of quality $\ell$, i.e.,   
$\Gamma^{(i, \ell)}$,  is given as
\begin{equation}
\Gamma^{(i,\ell)} = T_{i}^{(L_i,\ell)} - d_s - (L_i-1) \tau. \label{eq:base}
\end{equation}
In the next  section, we will use this stall time to determine the bound on the mean stall duration of the streamed video.

\section{Mean Stall Duration}
\label{sec:mean}
In this section, we will provide a bound on the mean stall duration for a file $i$. We will find the bound by two-stage probabilistic scheduling  and since this scheduling is one feasible  strategy, the obtained bound is an upper bound to the optimal strategy. Using \eqref{eq:base}, the expected stall time for file $i$ is given as follows
\begin{eqnarray}
\mathbb{E}\left[\Gamma^{(i,\ell)}\right] & = & \mathbb{E}\left[T_{i}^{(L_{i},\ell)}-d_{s}-\left(L_{i}-1\right)\tau\right]\nonumber \\
\nonumber \\
& = & \mathbb{E}\left[T_{i}^{(L_{i},\ell)}\right]-d_{s}-\left(L_{i}-1\right)\tau\label{eq:E_T_s_2}
\end{eqnarray}

Exact evaluation for the play time of segment $L_{i}$ is hard due to the dependencies between  $\mathcal{F}_{j,z,\nu_j,\ell}$ random variables for different values of $j$, $\nu_j$, $z$, and $\ell$, where $z\in{(1,2,...,L_{i}+1)}$ and $(j,\nu_j)\in \mathcal{A}_{i}^{(\ell)}$. Hence, we derive an upper-bound on the playtime of the segment $L_{i}$ as follows. Using Jensen's inequality \cite{kuczma2009introduction}, we have for $t_i>0$, 
 
 \begin{equation}
e^{t_{i}\mathbb{E}\left[T_{i}^{\left(L_{i},\ell\right)}\right]} \leq \mathbb{E}\left[e^{t_{i}T_{i}^{\left(L_{i},\ell\right)}}\right]. \label{eq:jensen}
 \end{equation}
 
 Thus, finding an upper bound on the moment generating function for $T_i^{(L_i,\ell)}$ can lead to an upper bound on the mean stall duration. Thus, we will now bound the moment generating function for $T_i^{(L_i,\ell)}$.

% Preview source code for paragraph 0
\begin{eqnarray}
 && \mathbb{E}\left[e^{t_{i}T_{i}^{\left(L_{i},\ell\right)}}\right] \overset{(a)}{=}  \mathbb{E}\left[\underset{z}{\mbox{max}}  \,\underset{(j,\nu_j)\in\mathcal{A}_{i}^{(\ell)}}{\mbox{max}}\,e^{t_{i}{\mathcal{F}}_{j,z,\nu_j,\ell}}\right]\\
 & = & \mathbb{E}_{\mathcal{A}_{i}^{(\ell)}}\left[\mathbb{E}\left[\underset{z}{\mbox{max}}\,\underset{(j,\nu_j)\in{\mathcal{A}_{i}}^{(\ell)}}{\mbox{max}}\,e^{t_{i}{\mathcal{F}}_{j,z,\nu_j,\ell}}|\,\mathcal{A}_{i}^{(\ell)}\right]\right]\\
 &\overset{(b)}{\leq} & \mathbb{E}_{\mathcal{A}_{i}^{(\ell)}}\left[\sum_{(j,\nu_j)\in \mathcal{A}_{i}^{(\ell)}}\mathbb{E}\left[\underset{z}{\mbox{max}}\,e^{t_{i}{\mathcal{F}}_{j,z,\nu_j,\ell}}\right]\right]
 \end{eqnarray}
 \begin{eqnarray}
 & = & \mathbb{E}_{\mathcal{A}_{i}^{(\ell)}}\left[\sum_{(j,\nu_j)}F_{i,j,\nu_j,\ell}\mathbf{1}_{\left\{ (j,\nu_j)\in\mathcal{A}_{i}^{(\ell)}\right\} }\right] \\
  & = & \sum_{(j,\nu_j)}F_{i,j,\nu_j,\ell}\,\mathbb{E}_{\mathcal{A}_{i}^{(\ell)}}\left[\mathbf{1}_{\left\{ (j,\nu_j)\in\mathcal{A}_{i}^{(\ell)}\right\} }\right]\\
 & = & \sum_{(j,\nu_j)}F_{i,j,\nu_j,\ell}\,\mathbb{P}\left((j,\nu_j)\in\mathcal{A}_{i}^{(\ell)}\right)\\
 & \overset{(c)}{=} & \sum_{j=1}^m \sum_{\nu_j=1}^{d_j}F_{i,j,\nu_j,\ell}\pi_{i,j,\nu_j}^{(\ell)}
\label{eq:mgf_bound}
\end{eqnarray}
where (a) follows from \eqref{eq:pjstart}, (b) follows by upper bounding $\max_{(j,\nu_j)\in \mathcal{A}_{i}^{(\ell)}} $ by $\sum_{(j,\nu_j)\in \mathcal{A}_{i}^{(\ell)}} $, (c) follows by two-stage probabilistic scheduling where $\mathbb{P}\left((j,\nu_j)\in\mathcal{A}_{i}^{(\ell)}\right) = \pi_{i,j,\nu_j}^{(\ell)}$,   and  $F_{i,j,\nu_j,\ell}\triangleq \mathbb{E}\left[\underset{z}{\mbox{max}}  \, e^{t_{i}\mathcal{F}_{i,z,\nu_j,\ell}}\right]$. We note that the only inequality here is for replacing the maximum by the sum. Since this term will be inside the logarithm for the mean stall latency, the gap between the term and its bound becomes additive rather than multiplicative.

To use  the bound \eqref{eq:mgf_bound}, $F_{i,j,\nu_j,\ell}$ needs to be bounded too. Thus, an upper bound on $F_{i,j,\nu_j,\ell}$  is calculated as follows.
\begin{align}
F_{i,j,\nu_j,\ell} & =\mathbb{E}\left[\underset{z}{\text{max}}\, e^{t_{i}\mathcal{F}_{j,z,\nu_j,\ell}}\right]\nonumber \\
 & \overset{(d)}{\leq}\sum_{z}\mathbb{E}\left[e^{t_{i}\mathcal{F}_{j,z,\nu_j,\ell}}\right]\nonumber \\
 & \overset{(e)}{=}e^{t_{i}(d_{s}+(L_{i}-1)\tau)}+\nonumber \\
 & \sum_{z=2}^{L_{i}+1}\frac{e^{t_{i}\left(L_{i}-z+1\right)\tau}\left(1-\rho_{j,\nu_j}\right)t_{i}B_{j,\nu_j}(t_{i})}{t_{i}-\Lambda_{j,\nu_j}\left(B_{j,\nu_j}(t_{i})-1\right)}\left(\frac{\alpha_{j,\nu_j}^{(\ell)}e^{t_{i}\beta_{j,\nu_j}^{(\ell)}}}{\alpha_{j,\nu_j}^{(\ell)}-t_{i}}\right)^{z-1}\nonumber \\
  & \overset{(f)}{=}e^{t_{i}(d_{s}+(L_{i}-1)\tau)}+\nonumber \\& \sum_{v=1}^{L_{i}}\frac{e^{t_{i}\left(L_{i}-v\right)\tau}\left(1-\rho_{j,\nu_j}\right)t_{i}B_{j,\nu_j}(t_{i})}{t_{i}-\Lambda_{j,\nu_j}\left(B_{j,\nu_j}(t_{i})-1\right)}\left(\frac{\alpha_{j,\nu_j}^{(\ell)} e^{t_{i}\beta_{j,\nu_j}^{(\ell)}}}{\alpha_{j,\nu_j}^{(\ell)}-t_{i}}\right)^{v}\label{eq:F_ij}
\end{align}
where (d) follows by bounding the maximum by the sum, (e) follows from (\ref{eq:momntPjz}), and (f) follows by substituting $v=z-1$. 

Substituting \eqref{eq:mgf_bound} in \eqref{eq:jensen}, we have 

\begin{equation}
\mathbb{E}\left[T_{i}^{(L_{i},\ell)}\right]\leq\frac{1}{t_{i}}\text{log}\left(\sum_{j=1}^{m}\sum_{\nu_j=1}^{d_j}\pi_{i,j,\nu_j}^{(\ell)}F_{i,j, \nu_j,\ell}\right).\label{eq:ET_i}
\end{equation}

Further, substituting the bounds \eqref{eq:F_ij} and \eqref{eq:ET_i} in \eqref{eq:E_T_s_2}, the  mean stall duration is bounded as follows.  

\begin{eqnarray}
 && \mathbb{E}\left[\Gamma^{(i,\ell)}\right]\leq \frac{1}{t_{i}}\text{log}\left(\sum_{j=1}^{m}\sum_{\nu_j=1}^{d_j}\pi_{i,j,\nu_j}^{(\ell)}\left(e^{t_{i}(d_{s}+(L_{i}-1)\tau)}+\right.\right.\nonumber\\
 && \left.\left.\sum_{v=1}^{L_{i}}e^{t_{i}\left(L_{i}-\nu\right)\tau}Z_{D_{{i,j,\nu_j}}}^{(v,\ell)}(t_{i})\right)\right)-\left(d_{s}+\left(L_{i}-1\right)\tau\right)\nonumber
 \end{eqnarray}
 \begin{eqnarray}
&&= \frac{1}{t_{i}}\text{log}\left(\sum_{j=1}^{m}\sum_{\nu_j=1}^{d_j}\pi_{i,j,\nu_j}^{(\ell)}\left(e^{t_{i}(d_{s}+(L_{i}-1)\tau)}+\right.\right.\nonumber\\
 && \left.\left.\sum_{v=1}^{L_{i}}e^{t_{i}\left(L_{i}-v\right)\tau}Z_{D_{{i,j,\nu_j}}}^{(v,\ell)}(t_{i})\right)\right)-\frac{1}{t_{i}}\text{log}\left(e^{t_{i}\left(d_{s}+\left(L_{i}-1\right)\tau\right)}\right) \nonumber\\
&= & \frac{1}{t_{i}}\text{log}\left(\sum_{j=1}^{m}\sum_{\nu_j=1}^{d_j}\pi_{i,j,\nu_j}^{(\ell)}\left(1+\right. \right. \nonumber\\
&& \left. \left.
\sum_{v=1}^{L_{i}}e^{-t_{i}\left(d_{s}+\left(v-1\right)\tau\right)}Z_{D_{i,j,\nu_j}}^{(v,\ell)}(t_{i})\right)\right),
\label{eq:ET_s_i}
\end{eqnarray}
where $Z_{D_{{i,j,\nu_j}}}^{(v,\ell)}(t_{i}) \triangleq \frac{\left(1-\rho_{j,\nu_j}\right)t_{i}B_{j,\nu_j}(t_{i})}{t_{i}-\Lambda_{j,\nu_j}\left(B_{j,\nu_j}(t_{i})-1\right)}\left(\frac{\alpha_{j,\nu_j}^{(\ell)} e^{t_{i}\beta_{j,\nu_j}^{(\ell)}}}{\alpha_{j,\nu_j}^{(\ell)}-t_{i}}\right)^{v}$. Let \\ $H_{i,j,\nu_j,\ell}=\sum_{v=1}^{L_{i}}e^{-t_{i}\left(d_{s}+\left(v-1\right)\tau\right)}Z_{D_{{i,j,\nu_j}}}^{(v,\ell)}(t_{i})$, which is the inner summation in  (\ref{eq:ET_s_i}). $H_{i,j,\nu_j,\ell}$ can be simplified using the geometric series formula to obtain

%as follows.
\begin{eqnarray}
  H_{i,j,\nu_{j},\ell} & = &\sum_{v=1}^{L_{i}}\left(\frac{e^{-t_{i}\left(d_{s}+\left(v-1\right)\tau\right)}\left(1-\rho_{j,\nu_{j}}\right)t_{i}B_{j,\nu_{j}}(t_{i})}{t_{i}-\Lambda_{j,\nu_{j}}\left(B_{j,\nu_{j}}(t_{i})-1\right)}\nonumber\right.\times\\
 && \left.\left(\frac{\alpha_{j,\nu_{j}}^{(\ell)}e^{t_{i}\beta_{j,\nu_{j}}^{(\ell)}}}{\alpha_{j,\nu_{j}}^{(\ell)}-t_{i}}\right)^{v}\right)\nonumber\\
 & =&\frac{e^{-t_{i}d_{s}}\left(1-\rho_{j,\nu_{j}}\right)t_{i}B_{j,\nu_{j}}(t_{i})}{t_{i}-\Lambda_{j,\nu_{j}}\left(B_{j,\nu_{j}}(t_{i})-1\right)}\times\nonumber\\
 && \sum_{v=1}^{L_{i}}\left(e^{-t_{i}\left(v-1\right)\tau}\left(\frac{\alpha_{j,\nu_{j}}^{(\ell)}e^{t_{i}\beta_{j,\nu_{j}}^{(\ell)}}}{\alpha_{j,\nu_{j}}^{(\ell)}-t_{i}}\right)^{v}\right)\nonumber\\
% & =\frac{e^{-t_{i}\left(d_{s}-\tau\right)}\left(1-\rho_{j,\nu_{j}}\right)t_{i}B_{j,\nu_{j}}(t_{i})}{t_{i}-\Lambda_{j,\nu_{j}}\left(B_{j,\nu_{j}}(t_{i})-1\right)}\times\nonumber\\
% & \sum_{v=1}^{L_{i}}\left(e^{-v\left(t_{i}\tau\right)}\left(\frac{\alpha_{j,\nu_{j}}^{(\ell)}e^{t_{i}\beta_{j,\nu_{j}}^{(\ell)}}}{\alpha_{j,\nu_{j}}^{(\ell)}-t_{i}}\right)^{v}\right)\nonumber\\
% & =\frac{e^{-t_{i}\left(d_{s}-\tau\right)}\left(1-\rho_{j,\nu_{j}}\right)t_{i}B_{j,\nu_{j}}(t_{i})}{t_{i}-\Lambda_{j,\nu_{j}}\left(B_{j,\nu_{j}}(t_{i})-1\right)}\times\nonumber\\
% & \sum_{v=1}^{L_{i}}\left(\frac{\alpha_{j,\nu_{j}}^{(\ell)}e^{t_{i}\beta_{j,\nu_{j}}^{(\ell)}-t_{i}\tau}}{\alpha_{j,\nu_{j}}^{(\ell)}-t_{i}}\right)^{v}\nonumber\\
% & =\frac{e^{-t_{i}\left(d_{s}-\tau\right)}\left(1-\rho_{j,\nu_{j}}\right)t_{i}B_{j,\nu_{j}}(t_{i})}{t_{i}-\Lambda_{j,\nu_{j}}\left(B_{j,\nu_{j}}(t_{i})-1\right)}\times\nonumber\\
% & \left(M_{j,\nu_{j}}^{(\ell)}(t_{\text{i}})e^{-t_{i}\tau}\frac{1-\left(M_{j,\nu_{j}}^{(\ell)}(t_{i})\right)^{Li}e^{-t_{i}L_{i}\tau}}{1-M_{j,\nu_{j}}^{(\ell)}(t_{i})e^{-t_{i}\tau}}\right)\nonumber\\
 & =&\frac{e^{-t_{i}\left(d_{s}-\tau\right)}\left(1-\rho_{j,\nu_{j}}\right)t_{i}B_{j,\nu_{j}}(t_{i})}{t_{i}-\Lambda_{j,\nu_{j}}\left(B_{j,\nu_{j}}(t_{i})-1\right)}\times\nonumber\\
 && \left[\frac{\widetilde{M}_{j,\nu_{j}}^{(\ell)}(t_{i})\left(1-\left(\widetilde{M}_{j,\nu_{j}}^{(\ell)}(t_{i})\right)^{L_{i}}\right)}{1-\widetilde{M}_{j,\nu_{j}}^{(\ell)}(t_{i})}\right]
\label{eq:H}\end{eqnarray}where \begin{equation}\widetilde{M}_{j,\nu_j}^{(\ell)}(t_{i})=M_{j,\nu_j}^{(\ell)}(t_{\text{i}})e^{-t_{i}\tau},\label{eq:tildem}\end{equation} $M_{j,\nu_j}^{(\ell)}(t_i)$ is given in \eqref{M_j_t_1}, and $B_{j,\nu_j}(t_i)$ is given in \eqref{eq:servTimeofFileB_j_i}. 
%Thus, the mean stall duration time for file $i$ is bounded by 

%\begin{equation}
%\mathbb{E}\left[T_{S}^{(i)}\right]\leq%\frac{1}{t_{i}}\text{log}\left(\sum_{j=1}^{m}\pi_{ij}\left(1+H_{ij}\right)\right)\label{eq:T_s_main_stall}
%\end{equation}

\begin{theorem}
The mean stall duration time for file $i$ streamed with quality  $\ell$ is bounded by 

\begin{equation}
\mathbb{E}\left[\Gamma^{(i,\ell)}\right]\leq\frac{1}{t_{i}}\text{log}\left(\sum_{j=1}^{m}\sum_{\nu_j=1}^{d_j} \pi_{i,j,\nu_j}^{(\ell)}\left(1+H_{i,j,\nu_j,\ell}\right)\right)\label{eq:T_s_main_stall}
\end{equation}
for any $t_{i}>0$, $\rho_{j,\nu_j}=\sum_{i,\ell}\pi_{i,j,\nu_j}^{(\ell)}\lambda_{i}b_{i,\ell}L_{i}\left(\beta_{j,\nu_j}^{(\ell)}+\frac{1}{\alpha_{j,\nu_j}^{(\ell)}}\right)$, $\rho_{j,\nu_j}<1,\,\text{and }$ \\
$\,\sum_{f=1}^r\sum_{\ell=1}^V\pi_{f,j,\nu_j}^{(\ell)}\lambda_{f}b_{f,\ell}\left(\frac{\alpha_{j,\nu_j}^{(\ell)}e^{-\beta_{j,\nu_j}^{(\ell)}t_{i}}}{\alpha_{j,\nu_j}^{(\ell)}-t_{i}}\right)^{L_{f}}-\left(\Lambda_{j,\nu_j}+t_{i}\right)<0,\,\forall j,\nu_j$.\label{meanthm}
\end{theorem}

Note that Theorem above holds only in the range of $t_i$ when  $t_i-\Lambda_{j,\nu_j}\left(B_{j,\nu_j}(t_i)-1\right)>0$ which reduces to \\
$\,\sum_{f=1}^r\sum_{\ell=1}^V\pi_{f,j,\nu_j}^{(\ell)}\lambda_{f}b_{f,\ell}\left(\frac{\alpha_{j,\nu_j}^{(\ell)}e^{-\beta_{j,\nu_j}^{(\ell)}t_{i}}}{\alpha_{j,\nu_j}^{(\ell)}-t_{i}}\right)^{L_{f}}-\left(\Lambda_{j,\nu_j}+t_{i}\right)<0,\,\forall i, j,\nu_j$,
and $\alpha_{j,\nu_{j}}-t_i>0$. Further, the server utilization $\rho_{j,\nu_j} $ must be less than $1$ for stability of the system.

%We note that for the scenario, where the files are downloaded rather than streamed, a metric of interest is the mean download time. This is a special case of our approach when the number of segments of each video is one, or $L_i=1$ and $\ell=1$. Thus, the mean download time of the file follows as a special case of Theorem \ref{meanthm}. We note that the authors of \cite{Xiang:2014:Sigmetrics:2014,Yu_TON} gave an upper bound for mean file download time using probabilistic scheduling. However, the bound in this paper is different since we use moment generating function based bound. The two bounds are compared in Section \ref{sec:num}, and the bound in this paper is shown to outperform that in \cite{Xiang:2014:Sigmetrics:2014,Yu_TON}.

\section{Optimization Problem Formulation and Proposed Algorithm} \label{sec:probForm}

\subsection{Problem Formulation}

Let $\boldsymbol{q}=(q_{i,j}^{(\ell)} \forall i=1,\ldots,r, j=1,\cdots,m, \ell=1,\ldots,V)$, $\boldsymbol{b}=(b_{i,\ell}, \forall i=1,\cdots,r, l=1, \cdots, V)$,  $\boldsymbol{w}=\left(w_{j,\nu_j} \forall j=1,\cdots,m, \nu_j = 1, \cdots, d_j\right)$,  $\boldsymbol{p}=\left(p_{j,\nu_j}^{(\ell)} \forall j=1,\cdots,m, \ \nu_j = 1, \cdots d_j, \ell = 1, \cdots, V \right)$, and \\
 $\boldsymbol{t}=\left(t_{1},t_{2},\cdots,t_{r}\right)$.  We wish to minimize the two proposed QoE metrics over the choice of two-stage probabilistic scheduling parameters, bandwidth allocation, probability of the quality of the streamed video and auxiliary variables. Since this is a multi-objective optimization, the objective can be modeled as a convex combination of the two QoE metrics.

Let $\overline{\lambda}=\sum_{i}\lambda_{i}$ be the total arrival rate of file $i$. Then, $\lambda_{i}/\overline{\lambda}$ is the ratio of video $i$ requests. The first objective is the minimization of the mean stall duration, averaged over all the file requests, and is given as $\sum_{i,\ell}\frac{\lambda_{i}}{\overline{\lambda}}\,\mathbb{E}\left[\Gamma^{\left(i,\ell\right)}\right]$. The second objective is  maximizing the streamed quality of all video requests, averaged over all the file requests, and is given as $\sum_{i,\ell}\frac{\lambda_i} {\overline{\lambda}} L_i b_{i,\ell} a_{\ell}$. Using the expressions for the mean stall duration in Section \ref{sec:mean} and the average streamed quality, optimization of a convex combination of the two QoE metrics can be formulated as follows.

\begin{align}
\ensuremath{\text{min}\,\,\,\sum_{i=1}^r\frac{\lambda_{i}}{\overline{\lambda}_{i}}\left[\theta\left(\sum_{\ell=1}^V-b_{i,\ell}L_{i}a_{\ell}\right)+\left(1-\theta\right)\right.\nonumber}\\
\left.\sum_{\ell}\frac{b_{i,\ell}}{t_{i}}\text{log}\left(\sum_{j=1}^m\sum_{\nu_{j}=1}^{d_j}q_{i,j}^{(\ell)}p_{j,\nu_{j}}^{(\ell)}(1+H_{i,j,\nu_{j},\ell})\right)\right]
\label{eq:joint_otp_prob}
\end{align}
s.t.  \eqref{eq:H}, \eqref{eq:tildem}, \eqref{M_j_t_1}, \eqref{eq:servTimeofFileB_j_i}, \eqref{eq:rho_j_1}, \eqref{eq:pi_i_j_nu}, \eqref{eq:alphabeta}, 
% \eqref{eq:H}, \eqref{eq:tildem}, \eqref{M_j_t_1}, \eqref{eq:servTimeofFileB_j_i}, \eqref{eq:rho_j_1}, \eqref{eq:pi_i_j_nu}, \eqref{eq:alphabeta}, 

	\begin{eqnarray}
	% &  & {H}_{i,j,\nu_j ,\ell}=\frac{e^{-{t}_{i}\left(d_{s}-\tau\right)}\left(1-\rho_{j,\nu_j}\right){t}_{i}B_{j,\nu_j}({t}_{i})}{{t}_{i}-\Lambda_{j,\nu_j}\left(B_{j,\nu_j}({t}_{i})-1\right)}{Q}_{i,j,\nu_j,\ell}\,\, ,
	%\label{eq:H_ij}\\
	% &  & {Q}_{i,j,\nu_j,\ell} =\left[\frac{{M}_{j,\nu_j}^{(\ell)}({t}_{i})\left(1-\left({M}_{j,\nu_j}^{(\ell)}({t}_{i})\right)^{L_{i}}\right)}{1-\widetilde{M}_{j,\nu_j}^{(\ell)}({t}_{i})}\right],\,\,
	%\label{eq:Q_ij2}\\ 
	%  &  & \widetilde{M}_{j,\nu_j}^{(\ell)}({t})=\frac{\alpha_{j,\nu_j}^{(\ell)}e^{\left(\beta_{j,\nu_j}^{(\ell)}-\tau\right){t}}}{\alpha_{j,\nu_j}^{(\ell)}-{t}},\,\, \label{eq:M_telda_opt2}\\ 
	%  &  & {B}_{j,\nu_j}(t)=\sum_{f=1}^{r} \sum_{\ell=1}^{V}\frac{\lambda_{f}q_{i,j}^{(\ell)}p_{j,\nu_{j}}b_{i,\ell}}{\Lambda_{j,\nu_j}}\left(\frac{\alpha_{j,\nu_j}^{(\ell)}e^{\beta_{j,\nu_j}^{(\ell)}{t}}}{\alpha_{j,\nu_j}^{(\ell)}-{t}}\right)^{L_f}, \, \forall j \label{eq:Bj_const}\\
&  & \rho_{j,\nu_j}<1  \ \ \    \forall j, \nu_j \label{eq:rho_j}\\
	&  & \varLambda_{j,\nu_j}=\sum_{f=1}^r\sum_{\ell=1}^{V}\lambda_{f} b_{f,\ell}q_{i,j}^{(\ell)}p_{j,\nu_{j}}^{(\ell)}\,\,\,\,\,\forall j, \nu_j\label{eq:Lambda_j}\\
%	&  & \sum_{j=1}^{m}q_{i,j}^{(\ell)}=k_{i}\, ,\,\forall i,\,\ell\,\,\label{eq:sum_ij}
&  & \sum_{j=1}^{m}q_{i,j}^{(\ell)}=k_{i}\, ,\,\forall i,\,\ell\,\,\label{eq:sum_ij}
		\end{eqnarray}
\begin{eqnarray}
	&  & \mbox{ \ensuremath{q_{i,j}^{(\ell)}}=0}\,\,\,\mbox{if \ensuremath{j\notin S_{i}^{(\ell)}}}\,,\ensuremath{q_{i,j}^{(\ell)}}\in\left[0,1\right] \label{eq:pij}\\
	&  & \sum_{\nu_j}p_{j,\nu_j}^{(\ell)} = 1,\,\, p_{j,\nu_j}^{(\ell)}\geq0,\,\,\forall j, \nu_j,\ell, 
	\label{eq:sum_p_nuj}\\
	% &  & \left|\mathcal{S}^{(\ell)}_{i}\right|=n_{i},\,\,\forall i, \ell\label{eq:S_i_and_ni}\\
		&  & \sum_{\ell}b_{i,\ell}=1,\,b_{i,\ell}\geq0,\,\,\forall i,\ell
	\label{eq:bil}\\
	&  & 0\leq w_{j,\nu_j} \leq 1,\,\,\forall j, \nu_j
	\label{eq:w_alpha}\\
		&  & \sum_{\nu_j}w_{j,\nu_j} \leq 1,\,\,\forall j, 
	\label{eq:sum_w_nuj}\\
	&  & 0<{t}_{i}<\alpha_{j,\nu_j}^{(\ell)},\,\forall i, j,\ell,\nu_j \label{eq:t_i_alpha_j2}\\
	&  & \alpha_{j,\nu_j}^{(\ell)}\left(e^{(\beta_{j,\nu_j}^{(\ell)}-\tau){t}_i}-1\right)+{t}_i<0\,,\forall i, j, \nu_j,  \ell \label{M_telda_less_1} \\
	&  & \sum_{f=1}^r\sum_{\ell=1}^{V}q_{f,j}^{(\ell)}p_{j,\nu_{j}}^{(\ell)}b_{f,\ell}\lambda_{f}\left(\frac{\alpha_{j,\nu_j}^{(\ell)}e^{\beta_{j,\nu_j}^{(\ell)}{t}_{i}}}{\alpha_{j,\nu_j}^{(\ell)}-{t}_{i}}\right)^{L_{f}}-\nonumber\\
	&&\left(\Lambda_{j,\nu_j}+{t}_{i}\right)<0,\,\forall i, j,\nu_j\label{eq:don_pos_cond2}\\
&  & 	\mbox{var.} \ \ \   \boldsymbol{q}, \boldsymbol{t}, \boldsymbol{b}, 	\boldsymbol{w}, 		\boldsymbol{p}
		\label{eq:vars}
		\end{eqnarray}

Here, $\theta\in [0,1]$ is a trade-off factor that determines the relative significance of the mean stall duration and the average streamed quality in the minimization problem. Varying $\theta=0$ to $\theta=1$, the solution for
(\ref{eq:joint_otp_prob}) spans the solutions that  maximize the video quality to those minimizing the mean stall duration. The equations \eqref{eq:H}, \eqref{eq:tildem}, \eqref{M_j_t_1}, \eqref{eq:servTimeofFileB_j_i}, \eqref{eq:rho_j_1}, \eqref{eq:pi_i_j_nu}, and \eqref{eq:alphabeta} give the terms in the objective function. The  constraint (\ref{eq:rho_j}) indicates that the  load intensity of server $j$ is less than $1$. Equation (\ref{eq:Lambda_j}) gives the aggregate arrival rate $\Lambda_j$ for each node. Constraints \eqref{eq:sum_ij}, \eqref{eq:pij}, and \eqref{eq:sum_p_nuj}  guarantee that the two-stage scheduling probabilities are feasible. Constraint \eqref{eq:bil}  guarantees that the quality assignment probabilities are feasible and \eqref{eq:sum_w_nuj} is for bandwidth splitting among different streams. Constraints \eqref{eq:t_i_alpha_j2}, (\ref{M_telda_less_1}), and (\ref{eq:don_pos_cond2})  ensure that $\widetilde{M}_{j}({t})$ and the moment generating function given in (\ref{eq:M_D_ij}) exist. In the next subsection, we will describe the proposed algorithm for this optimization problem.

%\begin{remark}
%
%\begin{remark}
%The  constraint (\ref{eq:don_pos_cond}) is convex with respect to $t_i$ and $\pi_{i,j}$ individually
%\end{remark}

%$\boldsymbol{\pi=}\left(\pi_{ij},\forall %i=1,2,3,\ldots,r;j=1,2,\ldots,m\right)$
%and $\boldsymbol{t}=\left(t_{1},t_{2},\ldots,t_{r}\right)$
%\end{remark}

%\begin{lemma}
%The proposed optimization problem (\ref{eq:joint_otp_prob}) is non-convex, in  $\boldsymbol{\pi=}\left(\pi_{ij},\forall i=1,2,3,\ldots,r;j=1,2,\ldots,m\right)$  given 
% $\boldsymbol{t}=\left(t_{1},t_{2},\ldots,t_{r}\right)$
%\end{lemma}
%
%\begin{proof}
%The second term includes $log$ function with its  
%
%\end{proof}
%
%
%\begin{lemma}
%The proposed optimization problem (\ref{eq:joint_otp_prob}) is non-convex, in 
% $\boldsymbol{t}=\left(t_{1},t_{2},\ldots,t_{r}\right)$
% given 
% $\boldsymbol{\pi=}\left(\pi_{ij},\forall i=1,2,3,\ldots,r;j=1,2,\ldots,m\right)$ 
%\end{lemma}

%\begin{proof}

%\end{proof}

%\textbf{Remark 2}. The problem is non-convex jointly in $\boldsymbol{t}$
%and $\boldsymbol{\pi}$ since both two terms of the objective function include the product of  $\boldsymbol{t}$
%and $\boldsymbol{\pi}$.

%Note that the proposed  optimization problem is a mixed integer optimization
%as we have the placement over $n$ servers. 

\subsection{Proposed Algorithm}\label{sec:algo}

The mean stall duration optimization problem given in  \eqref{eq:joint_otp_prob}-\eqref{eq:vars}  is optimized over five set of variables:  server scheduling probabilities $\boldsymbol{q}$, PS selection probabilities $\boldsymbol{p}$,  auxiliary parameters $\boldsymbol{t}$, video quality parameters $\boldsymbol{b}$, and bandwidth allocation weights $\boldsymbol{w}$. We first note that the problem is non-convex in all the parameters jointly, which can be easily seen in the terms which are product of the different variables. Since the problem is non-convex, we propose an iterative algorithm to solve the problem. The proposed algorithm divides the problem into five sub-problems that optimize one variable while fixing the remaining four. The five sub-problems are labeled as   (i) Server Access Optimization: optimizes  $\boldsymbol{q}$, for given $\boldsymbol{p}$, $\boldsymbol{t}$, $\boldsymbol{b}$ and $\boldsymbol{w}$,  (ii) PS Selection Optimization: optimizes  $\boldsymbol{p}$, for given $\boldsymbol{q}$, $\boldsymbol{t}$, $\boldsymbol{b}$ and $\boldsymbol{w}$,  (iii) Auxiliary Variables Optimization: optimizes  $\boldsymbol{t}$ for given $\boldsymbol{q}$, $\boldsymbol{p}$,  $\boldsymbol{b}$ and $\boldsymbol{w}$,  and (iv) Video Quality Optimization: optimizes  $\boldsymbol{b}$ for given $\boldsymbol{q}$, $\boldsymbol{p}$, $\boldsymbol{t}$, and $\boldsymbol{w}$, and (v) Bandwidth Allocation Optimization: optimizes  $\boldsymbol{w}$ for given $\boldsymbol{q}$, $\boldsymbol{p}$, $\boldsymbol{t}$, and $\boldsymbol{b}$. The algorithm is summarized as follows. 

%, where $\boldsymbol{\mathcal{S}}=\left(\mathcal{S}_{1},\mathcal{S}_{2},\ldots,\mathcal{S}_{r}\right)$,
% (ii) access-scheduling probabilities  ($\boldsymbol{\pi}$- Optimization ), and finally (iii) auxiliary variables optimization ($\boldsymbol{t}$- Optimization). These sub-problems have an easier-to-handle structure and can be solved efficiently. Next, we develop an efficient algorithm solution for our optimization problem.
%
%
%The joint mean-tail stall duration optimization problem given in (\ref{eq:joint_otp_prob}) is carried out over three set of variables: chunk placement $\mathcal{S}_{i}$, $\forall$ $i$, scheduling probabilities $\pi_{i,j}$, $\forall$ $i,j$, and auxiliary parameter $t_{i}$, $\forall$ $i$. In order to solve the non-convex problem, we develop an alternating minimization algorithm which is shown to converge to a stationary solution. The proposed algorithm can be written as follows.

\begin{enumerate}[leftmargin=0cm,itemindent=.5cm,labelwidth=\itemindent,labelsep=0cm,align=left]
\item \textbf{Initialization}: Initialize $\boldsymbol{t}$,\,$\boldsymbol{b}$, $\boldsymbol{w}$, $\boldsymbol{p}$,
and $\boldsymbol{q}$ in the feasible set. 
\item \textbf{While Objective Converges}
\begin{enumerate}

\item  Run Server Access Optimization using current values of $\boldsymbol{p}$, $\boldsymbol{t}$, $\boldsymbol{b}$,
and $\boldsymbol{w}$ to get new values of $\boldsymbol{q}$
\item Run PS Selection Optimization using current values of $\boldsymbol{q}$, $\boldsymbol{t}$, $\boldsymbol{b}$,
and $\boldsymbol{w}$ to get new values of $\boldsymbol{p}$
\item  Run Auxiliary Variables Optimization using current values of $\boldsymbol{q}$, $\boldsymbol{p}$, $\boldsymbol{b}$,
and $\boldsymbol{w}$ to get new values of $\boldsymbol{t}$
\item  Run Streamed Quality   Optimization using current values
of $\boldsymbol{q}$,  $\boldsymbol{p}$,  $\boldsymbol{t}$, and  $\boldsymbol{w}$ to get new values of $\boldsymbol{b}$.
\item  Run Bandwidth Allocation Optimization using current values
of $\boldsymbol{q}$, $\boldsymbol{p}$, $\boldsymbol{t}$, and $\boldsymbol{b}$ to get new values of $\boldsymbol{w}$. 
\end{enumerate}
\end{enumerate}

We next describe the five sub-problems along with the proposed solutions for the sub-problems.

\subsubsection{Server Access Optimization}

Given the probability distribution of the streamed video quality, the bandwidth allocation weights, the PS selection probabilities, and the auxiliary variables, this subproblem can be written as follows. 

\textbf{Input: $\boldsymbol{t}$, $\boldsymbol{b}$, $\boldsymbol{p}$, and  $\boldsymbol{w}$}

\textbf{Objective:} $\qquad\quad\;\;$min $\left(\ref{eq:joint_otp_prob}\right)$

$\hphantom{\boldsymbol{\text{Objective:}\,}}\qquad\qquad$s.t. \eqref{eq:rho_j}, \eqref{eq:Lambda_j}, \eqref{eq:sum_ij}, 
\eqref{eq:pij},  \eqref{eq:don_pos_cond2}

$\hphantom{\boldsymbol{\text{Objective:}\,}}\qquad\qquad$var. $\boldsymbol{q}$

%Now, we explain our approach in solving $\boldsymbol{\pi}$- Optimization sub-problem. We optimize over $\pi_{ij}$, $\forall i,j$ for given placement $\boldsymbol{\mathcal{S}}$ and fixed $\boldsymbol{t}$. This problem is non-convex as shown in Section \ref{sec:probForm}. Hence, 

In order to solve this problem, we have used iNner cOnVex
Approximation (NOVA)  algorithm proposed in \cite{scutNOVA} to solve this sub-problem. The key idea for this algorithm is that the non-convex objective function is replaced by suitable convex approximations at which convergence to a stationary solution of the original non-convex optimization is established. NOVA solves the approximated function efficiently and maintains feasibility in each iteration. The objective function can be approximated by a convex one (e.g., proximal gradient-like approximation) such that the first order properties are preserved \cite{scutNOVA},  and this convex approximation can be used in NOVA algorithm.

Let $\widetilde{U_q}\left(\boldsymbol{q};\boldsymbol{q^\nu}\right)$ be the
convex approximation at iterate $\boldsymbol{q^\nu}$ to the original non-convex problem $U\left(\boldsymbol{q}\right)$, where $U\left(\boldsymbol{q}\right)$ is given by (\ref{eq:joint_otp_prob}). Then, a valid choice of  $\widetilde{U_q}\left(\boldsymbol{q};\boldsymbol{q^\nu}\right)$ is the first order approximation of  $U\left(\boldsymbol{q}\right)$, e.g., (proximal) gradient-like approximation, i.e.,  
\begin{equation}
\widetilde{U_q}\left(\boldsymbol{q},\boldsymbol{q^\nu}\right)=\nabla_{\boldsymbol{q}}U\left(\boldsymbol{q^\nu}\right)^{T}\left(\boldsymbol{q}-\boldsymbol{q^\nu}\right)+\frac{\tau_{u}}{2}\left\Vert \boldsymbol{q}-\boldsymbol{q^\nu}\right\Vert ^{2},\label{eq:U_x_u_bar}
\end{equation}
where $\tau_u$ is a regularization parameter. Note that all the constraints \eqref{eq:rho_j}, \eqref{eq:Lambda_j}, \eqref{eq:sum_ij},  \eqref{eq:pij},  and \eqref{eq:don_pos_cond2} are linear in $\boldsymbol{q_{i,j}}$.  The NOVA Algorithm for optimizing $\boldsymbol{q}$ is described in  Algorithm \ref{alg:NOVA_Alg1Pi} (given in Appendix \ref{apdx_table}). Using the convex approximation $\widetilde{U_q}\left(\boldsymbol{q};\boldsymbol{q^\nu}\right)$, the minimization steps in Algorithm \ref{alg:NOVA_Alg1Pi} are convex, with linear constraints and thus can be solved using a projected gradient descent algorithm. A step-size ($\gamma
$)  is also used in the update of the iterate $\boldsymbol{q}^{\nu}$. Note that the iterates $\left\{ \boldsymbol{q}^{(\nu)}\right\} $ generated by the algorithm are all feasible for the original problem and, further, convergence is guaranteed, as shown in \cite{scutNOVA} and described in   lemma \ref{lem_pi}. %For further reading on NOVA algorithm, we refer the interested reader to \cite{scutNOVA} for detailed 

In order to use NOVA, there are some assumptions (given in \cite{scutNOVA}) that have to be satisfied in both original function and its approximation. These assumptions can be classified into two categories. The first category is the set of conditions that ensure that the original problem and its constraints are continuously differentiable on the domain of the function, which are satisfied in our problem. The second category is the set of conditions that ensures that the approximation of the original problem is uniformly strongly convex on the domain of the function. The latter set of conditions are also satisfied as the chosen function is strongly convex and its domain is also convex. To see this, we need to show that the constraints \eqref{eq:rho_j}, \eqref{eq:Lambda_j}, \eqref{eq:sum_ij}, 
\eqref{eq:pij},  \eqref{eq:don_pos_cond2} form a convex domain in $\mathbf{q}$ which is easy to see from the linearity of the constraints in $\mathbf{q}$. Further details on the assumptions and function approximation can be found in \cite{scutNOVA}. Thus, the following result holds.

\begin{lemma}  \label{lem_pi}
For fixed $\boldsymbol{b}$, $\boldsymbol{p}$, $\boldsymbol{w}$, and  $\boldsymbol{t}$, the optimization of our problem over 
$\boldsymbol{q}$ generates a sequence of  decreasing
objective values and therefore is guaranteed to converge to a stationary point.
\end{lemma}
%Since NOVA Algorithm generates a sequence of monotonically decreasing
%objective values and therefore is guaranteed to converge to a stationary point \cite{scutNOVA}. Hence, the above Lemma follows.  

\subsubsection{Auxiliary Variables Optimization}
Given the probability distribution of the streamed video quality, the bandwidth allocation weights, the PS selection probabilitites and the server scheduling probabilities, this subproblem can be written as follows. 

\textbf{Input: $\boldsymbol{q}$, $\boldsymbol{p}$, $\boldsymbol{b}$, and $\boldsymbol{w}$ }

\textbf{Objective:} $\qquad\quad\;\;$min $\left(\ref{eq:joint_otp_prob}\right)$

$\hphantom{\boldsymbol{\text{Objective:}\,}}\qquad\qquad$s.t.  \eqref{eq:t_i_alpha_j2}, \eqref{M_telda_less_1}, 
 \eqref{eq:don_pos_cond2}

$\hphantom{\boldsymbol{\text{Objective:}\,}}\qquad\qquad$var. $\boldsymbol{t}$

Similar to Access Optimization, this optimization can be solved using NOVA algorithm. The constraint  \eqref{eq:t_i_alpha_j2} is linear in $\boldsymbol{t}$. Further, the next two Lemmas show that the constraints  \eqref{M_telda_less_1} and  \eqref{eq:don_pos_cond2} are convex in $\boldsymbol{t}$, respectively. 

\begin{lemma}
The  constraint \eqref{M_telda_less_1} is convex with respect to  $\boldsymbol{{t}}$. 
\end{lemma}

\begin{proof}

The constraint \eqref{M_telda_less_1} is separable for each $t_i$ and thus it is enough to prove convexity of $C(t)=\alpha_{j,\nu_j}\left(e^{\left(\beta_{j,\nu_j}-\tau\right)t}-1\right)+t$. Thus, it is enough to prove that $C''(t)\ge 0$. 

The first derivative of $C(t)$ is given as

\begin{equation}
C'(t)=\alpha_{j,\nu_j}\left(\left(\beta_{j,\nu_j}-\tau\right)e^{\left(\beta_{j,\nu_j}-\tau\right)t}\right)+1
\end{equation}

Differentiating it again, we get the second derivative as follows. 
\begin{equation}
C''(t)  =\alpha_{j,\nu_j}\left(\beta_{j,\nu_j}-\tau\right)^{2}e^{\left(\beta_{j,\nu_j}-\tau\right)t}\label{eq:C_22}
\end{equation}

Since $\alpha_{j,\nu_j}\ge0$, $C''(t)$ given in  (\ref{eq:C_22}) is non-negative, which proves the Lemma. 
\vspace{-.1in}
\end{proof}

\begin{lemma}
The  constraint \eqref{eq:don_pos_cond2} is convex with respect to   $\boldsymbol{t}$.  
\end{lemma}

\begin{proof}

The constraint \eqref{eq:don_pos_cond2}   is separable for each   $t_i$, and thus it is enough to prove convexity of\\ $E(t) = \sum_{f=1}^r\pi_{f,j,\nu_j}\lambda_{f}b_{f,\ell}a_{\ell}\left(\frac{\alpha_{j,\nu_j}e^{\beta_{j,\nu_j}t}}{\alpha_{j,\nu_j}-t}\right)^{L_{f}}-\left(\Lambda_{j,\nu_j}+t\right)$ for $t<\alpha_{j,\nu_j}$. Thus, it is enough to prove that $E''(t)\ge 0$  for $t<\alpha_{j,\nu_j}$. We further note that it is enough to prove that $D''(t)\ge 0$, where $D(t) = \frac{e^{L_f\beta_{j,\nu_j}t}}{(\alpha_{j,\nu_j}-t)^{L_{f}}}$. This follows since 
\begin{eqnarray}
D^{'}(t) & =&\frac{L_{f}e^{L_{f}\beta_{j,\nu_j}t}\left[\beta_{j,\nu_j}+\left(\alpha_{j,\nu_j}-t\right)^{-1}\right]}{\left(\alpha_{j,\nu_j}-t\right)^{L_{f}}}\ge 0\nonumber\\
D^{''}(t)& =&\frac{L_{f}\beta_{j,\nu_j}e^{L_{f}\beta_{j,\nu_j}t}\left[\beta_{j,\nu_j}+\frac{1+L_{f}}{\alpha_{j,\nu_j}-t}\left(1+\frac{1/\beta_{j,\nu_j}}{\alpha_{j,\nu_j}-t}\right)\right]}{\left(\alpha_{j,\nu_j}-t\right)^{L_{f}+2}}\ge 0\nonumber
\end{eqnarray}

%The first derivative of $D(t)$ is given as

%\begin{equation}
%	D'(t) =\frac{\pi_{ij}\lambda_{i}L_{i}\alpha_{j,\nu_j}^{L_i}e^{\beta_{j,\nu_j}tL_{i}}}{\left(\alpha_{j,\nu_j}-t\right)^{L_i+1}}\left[\beta_{j,\nu_j}\left(\alpha_{j,\nu_j}-t\right)+1\right]
%\end{equation}

%Note that $E'(t)\ge0$ since $\alpha_{j,\nu_j}\get$.  Differentiating it again to get the second derivative, we get the second derivative as follows. 
%\begin{equation}
%	E''(t)=\left(1+\frac{1}{\alpha_{j,\nu_j}-t}\right)\left((\alpha_{j,\nu_j}-t)\,L_{i}+L_{i}+1\right)-1 
%	\label{eq:E_22}
%\end{equation}

%Since $\alpha_{j,\nu_j}\get$, $E''(t)$ given in  \eqref{eq:E_22} is non-negative, which proves the Lemma. 
\vspace{-.1in}
\end{proof}

%(given in Appendix \ref{apdx_table})

 Algorithm \ref{alg:NOVA_Alg1} (given in Appendix \ref{apdx_table}) shows the used procedure to solve for $\boldsymbol{t}$. Let $\overline{U}\left(\boldsymbol{t};\boldsymbol{t^\nu}\right)$ be the
convex approximation at iterate $\boldsymbol{t^\nu}$ to the original non-convex problem $U\left(\boldsymbol{t}\right)$, where $U\left(\boldsymbol{t}\right)$ is given by (\ref{eq:joint_otp_prob}), assuming other parameters constant. Then, a valid choice of  $\overline{U}\left(\boldsymbol{t};\boldsymbol{t^\nu}\right)$ is the first order approximation of  $U\left(\boldsymbol{t}\right)$, i.e.,  
\begin{equation}
\overline{U}\left(\boldsymbol{t},\boldsymbol{t^\nu}\right)=\nabla_{\boldsymbol{t}}U\left(\boldsymbol{t^\nu}\right)^{T}\left(\boldsymbol{t}-\boldsymbol{t^\nu}\right)+\frac{\tau_{t}}{2}\left\Vert \boldsymbol{t}-\boldsymbol{t^\nu}\right\Vert ^{2}.\label{eq:U_t_u_bar}
\end{equation}
where $\tau_t$ is a regularization parameter. The detailed steps can be seen in Algorithm \ref{alg:NOVA_Alg1}. Since all the constraints \eqref{eq:t_i_alpha_j2}, \eqref{M_telda_less_1}, and  \eqref{eq:don_pos_cond2} have been shown to be convex in $\boldsymbol{t}$, the optimization problem in Step 1 of  Algorithm \ref{alg:NOVA_Alg1} can be solved by the standard projected gradient descent algorithm. 

%Similar to Access  Optimization problem, we get the minimizer for (\ref{eq:U_t_u_bar}), step 1 in Algorithm \ref{alg:NOVA_Alg1}), and then continue until convergence. 

%In this sub-problem, we solve for all $t_i$, for fixed placement and given accessing probabilities. Similar to $\boldsymbol{\pi}$- Optimization problem, we used NOVA algorithm to solve for $t$.  Algorithm \ref{alg:NOVA_Alg1}, shows the used procedure to solve for $\boldsymbol{t}$. After initialization with feasible choice, NOVA keeps generating  a sequence of monotonically decreasing $\left\{ \boldsymbol{t}^{\nu+1}\right\} $ until  convergence.

\begin{lemma} \label{lem_t}

For fixed $\boldsymbol{q}$, $\boldsymbol{b}$, $\boldsymbol{w}$, and  $\boldsymbol{p}$, the optimization of our problem over 
$\boldsymbol{t}$ generates a sequence of monotonically decreasing
objective values and therefore is guaranteed to converge to a stationary point.
\end{lemma}
%\begin{proof}
%Since NOVA Algorithm generates a sequence of monotonically decreasing
%objective values as shown in \cite{scutNOVA}, $\boldsymbol{t}$ is guaranteed to converge to a stationary point.
%\end{proof}

%\input{plc}

\subsubsection{Streamed Video Quality Optimization}
Given the auxiliary variables, the bandwidth allocation weights, the PS  selection probabilities, and the scheduling probabilities, this subproblem can be written as follows. 

\textbf{Input: $\boldsymbol{q}$, $\boldsymbol{p}$, $\boldsymbol{t}$, and $\boldsymbol{w}$}

\textbf{Objective:} $\qquad\quad\;\;$min $\left(\ref{eq:joint_otp_prob}\right)$

$\hphantom{\boldsymbol{\text{Objective:}\,}}\qquad\qquad$s.t. 
\eqref{eq:rho_j},
\eqref{eq:Lambda_j},
\eqref{eq:bil},
\eqref{eq:don_pos_cond2}

$\hphantom{\boldsymbol{\text{Objective:}\,}}\qquad\qquad$var. $\boldsymbol{b}$

Similar to the aforementioned  two Optimization problems, this optimization can be solved using NOVA algorithm. The constraints  \eqref{eq:rho_j}, \eqref{eq:Lambda_j}, \eqref{eq:bil}, and \eqref{eq:don_pos_cond2} are linear in $\boldsymbol{b}$, and hence, form a convex domain.

 Algorithm \ref{alg:NOVA_Alg3} (given in Appendix \ref{apdx_table}) shows the used procedure to solve for $\boldsymbol{b}$. Let ${U_b}\left(\boldsymbol{b};\boldsymbol{b^\nu}\right)$ be the
convex approximation at iterate $\boldsymbol{b^\nu}$ to the original non-convex problem $U\left(\boldsymbol{b}\right)$, where $U\left(\boldsymbol{b}\right)$ is given by (\ref{eq:joint_otp_prob}), assuming other parameters constant. Then, a valid choice of  ${U_b}\left(\boldsymbol{b};\boldsymbol{b^\nu}\right)$ is the first order approximation of  $U\left(\boldsymbol{b}\right)$, i.e.,  
\begin{equation}
{U_b}\left(\boldsymbol{b},\boldsymbol{b^\nu}\right)=\nabla_{\boldsymbol{b}}U\left(\boldsymbol{b^\nu}\right)^{T}\left(\boldsymbol{b}-\boldsymbol{b^\nu}\right)+\frac{\tau_{b}}{2}\left\Vert \boldsymbol{b}-\boldsymbol{b^\nu}\right\Vert ^{2}.\label{eq:U_b}
\end{equation}
where $\tau_t$ is a regularization parameter. The detailed steps can be seen in Algorithm \ref{alg:NOVA_Alg3}. Since all the constraints  have been shown to be convex in $\boldsymbol{b}$, the optimization problem in Step 1 of  Algorithm \ref{alg:NOVA_Alg3} can be solved by the standard projected gradient descent algorithm. 

%Similar to Access  Optimization problem, we get the minimizer for (\ref{eq:U_t_u_bar}), step 1 in Algorithm \ref{alg:NOVA_Alg1}), and then continue until convergence. 

%In this sub-problem, we solve for all $t_i$, for fixed placement and given accessing probabilities. Similar to $\boldsymbol{\pi}$- Optimization problem, we used NOVA algorithm to solve for $t$.  Algorithm \ref{alg:NOVA_Alg1}, shows the used procedure to solve for $\boldsymbol{t}$. After initialization with feasible choice, NOVA keeps generating  a sequence of monotonically decreasing $\left\{ \boldsymbol{t}^{\nu+1}\right\} $ until  convergence.

\begin{lemma} \label{lem_b}

For fixed $\boldsymbol{t}$, $\boldsymbol{w}$, $\boldsymbol{p}$, and  $\boldsymbol{q}$, the optimization of our problem over 
$\boldsymbol{b}$ generates a sequence of monotonically decreasing
objective values and therefore is guaranteed to converge to a stationary point.
\end{lemma}
%\begin{proof}
%Since NOVA Algorithm generates a sequence of monotonically decreasing
%objective values as shown in \cite{scutNOVA}, $\boldsymbol{t}$ is guaranteed to converge to a stationary point.
%\end{proof}

%\input{apdx_table}
\subsubsection{ Bandwidth Allocation Weights Optimization}
Given the auxiliary variables, the streamed video quality probabilities, the PS selection probabilities, and the scheduling probabilities, this subproblem can be written as follows. 

\textbf{Input: $\boldsymbol{q}$, $\boldsymbol{p}$, $\boldsymbol{t}$, and $\boldsymbol{b}$}

\textbf{Objective:} $\qquad\quad\;\;$min $\left(\ref{eq:joint_otp_prob}\right)$

$\hphantom{\boldsymbol{\text{Objective:}\,}}\qquad\qquad$s.t. 
\eqref{eq:rho_j},
\eqref{eq:w_alpha},
\eqref{eq:sum_w_nuj},
\eqref{eq:don_pos_cond2}

$\hphantom{\boldsymbol{\text{Objective:}\,}}\qquad\qquad$var. $\boldsymbol{w}$

This optimization problem can be solved using NOVA algorithm. It is easy to notice that the constraints    \eqref{eq:w_alpha} and \eqref{eq:sum_w_nuj} are linear and thus  convex with respect to $\boldsymbol{b}$. Further, the next two Lemmas show that the constraints  \eqref{eq:rho_j} and  \eqref{eq:don_pos_cond2} are convex in $\boldsymbol{w}$, respectively. 

\begin{lemma}
The  constraint \eqref{eq:rho_j} is convex with respect to  $\boldsymbol{{w}}$. 
\end{lemma}

\begin{proof}

Since there is no coupling between the subscripts $j$, $\ell$, and $\nu_j$ in \eqref{eq:rho_j}, we remove the subscripts in the rest of the proof. Moreover, since $\alpha$ is linear in $w$, it is enough to prove the convexity with respect to $\alpha$. Also, the constraint \eqref{eq:rho_j} is separable for each $\alpha$ and thus it is enough to prove convexity of $C_1(\alpha) = 1/\alpha$. It is easy to show that the second derivative of $C_1(\alpha)$ with respect to $\alpha$ is given by

%Thus, it is enough to prove that $C_1''(\alpha)\ge 0$. 
%
%The first derivative of $C_1(\alpha)$ is given as
%
%\begin{equation}
%C_1(\alpha)=\frac{-1}{\alpha^2}
%\end{equation}
%
%Differentiating it again, we get the second derivative as follows. 
\begin{equation}
C_1^{''}(\alpha)=\frac{2}{\alpha^3}\label{eq:C_alpha}
\end{equation}

Since $\alpha\ge0$, $C_1^{''}(\alpha)$ given in  (\ref{eq:C_alpha}) is non-negative, which proves the Lemma. 
\vspace{-.1in}
\end{proof}

\begin{lemma}
The  constraint \eqref{eq:don_pos_cond2} is convex with respect to   $\boldsymbol{{w}}$.  
\end{lemma}

\begin{proof}

The constraint \eqref{eq:don_pos_cond2}   is separable for each  $\alpha_{j,\nu_j}^{\ell}$, and thus it is enough to prove convexity of \\
$E_1(\alpha_{j,\nu_j}^{(\ell)}) = \sum_{f=1}^r \sum_{\ell=1}^V \pi_{f,j,\nu_j}^{(\ell)}\lambda_{f}b_{f,\ell}\left(\frac{\alpha_{j,\nu_j}^{(\ell)}e^{\beta_{j,\nu_j}^{(\ell)}t}}{\alpha_{j,\nu_j}^{(\ell)}-t}\right)^{L_{f}}-\left(\Lambda_{j,\nu_j}+t\right)$ for $t<\alpha_{j,\nu_j}^{(\ell)}$. Since there is only a single index $j$, $\nu_j$, and $\ell$ here, we ignore the subscripts and superscripts for the
rest of this proof. Thus, it is enough to prove that $E_1''(\alpha)\ge 0$  for $t<\alpha$. We further note that it is enough to prove that $D_1''(\alpha)\ge 0$, where $\ensuremath{D_{1}(\alpha)=\left(1-\frac{t}{\alpha}\right)^{-L_{i}}}$. This holds since,

\begin{align}
D_{1}^{'}(\alpha) & =\frac{-L_{i}\times t}{\alpha^{2}}\left(\frac{\alpha}{\alpha-t}\right)^{L_{i}+1}\label{eq:D_1_alpha}\\
D_{1}^{''}(\alpha) & =\frac{L_{i}\times t}{\alpha^{3}}\left(\frac{\alpha}{\alpha-t}\right)^{L_{i}+1}\left[2+\frac{\alpha\left(L_{i}+1\right)}{\alpha_{j}-t}\right]\ge0\label{eq:D_2_alpha}
\end{align}
\vspace{-.1in}
\end{proof}

Algorithm \ref{alg:NOVA_Alg4} (given in Appendix \ref{apdx_table})  shows the used procedure to solve for $\boldsymbol{w}$. Let ${U_w}\left(\boldsymbol{w};\boldsymbol{w^\nu}\right)$ be the
convex approximation at iterate $\boldsymbol{w^\nu}$ to the original non-convex problem $U\left(\boldsymbol{w}\right)$, where $U\left(\boldsymbol{w}\right)$ is given by (\ref{eq:joint_otp_prob}), assuming other parameters constant. Then, a valid choice of  ${U_w}\left(\boldsymbol{w};\boldsymbol{w^\nu}\right)$ is the first order approximation of  $U\left(\boldsymbol{w}\right)$, i.e.,  
\begin{equation}
{U_w}\left(\boldsymbol{w},\boldsymbol{w^\nu}\right)=\nabla_{\boldsymbol{w}}U\left(\boldsymbol{w^\nu}\right)^{T}\left(\boldsymbol{w}-\boldsymbol{w^\nu}\right)+\frac{\tau_{w}}{2}\left\Vert \boldsymbol{w}-\boldsymbol{w^\nu}\right\Vert ^{2}.\label{eq:U_b}
\end{equation}
where $\tau_t$ is a regularization parameter. The detailed steps can be seen in Algorithm \ref{alg:NOVA_Alg4}. Since all the constraints  have been shown to be convex, the optimization problem in Step 1 of  Algorithm \ref{alg:NOVA_Alg4} can be solved by the standard projected gradient descent algorithm.

\begin{lemma} \label{lem_t}

For fixed $\boldsymbol{q}$, $\boldsymbol{p}$, $\boldsymbol{t}$, and  $\boldsymbol{b}$, the optimization of our problem over 
$\boldsymbol{w}$ generates a sequence of  decreasing
objective values and therefore is guaranteed to converge to a stationary point.
\end{lemma}
%\begin{proof}
%Since NOVA Algorithm generates a sequence of monotonically decreasing
%objective values as shown in \cite{scutNOVA}, $\boldsymbol{t}$ is guaranteed to converge to a stationary point.
%\end{proof}

\subsubsection{PS Selection Probabilities}
Given the auxiliary variables, the bandwidth allocation weights, the streamed video quality probabilities, and the scheduling probabilities, this subproblem can be written as follows. 

\textbf{Input: $\boldsymbol{q}$, $\boldsymbol{b}$, $\boldsymbol{t}$, and $\boldsymbol{w}$}

\textbf{Objective:} $\qquad\quad\;\;$min $\left(\ref{eq:joint_otp_prob}\right)$

$\hphantom{\boldsymbol{\text{Objective:}\,}}\qquad\qquad$s.t. 
\eqref{eq:rho_j},
\eqref{eq:Lambda_j},
\eqref{eq:sum_p_nuj},
\eqref{eq:don_pos_cond2},

$\hphantom{\boldsymbol{\text{Objective:}\,}}\qquad\qquad$var. $\boldsymbol{p}$

 This optimization can be solved using NOVA algorithm. The constraints  \eqref{eq:rho_j}, \eqref{eq:Lambda_j}, \eqref{eq:sum_p_nuj}, and \eqref{eq:don_pos_cond2} are linear in $\boldsymbol{p}$, and hence, the domain is convex. 

 Algorithm \ref{alg:NOVA_Alg5} (given in Appendix \ref{apdx_table}) shows the used procedure to solve for $\boldsymbol{p}$. Let ${U_p}\left(\boldsymbol{p};\boldsymbol{p^\nu}\right)$ be the
convex approximation at iterate $\boldsymbol{p^\nu}$ to the original non-convex problem $U\left(\boldsymbol{p}\right)$, where $U\left(\boldsymbol{p}\right)$ is given by (\ref{eq:joint_otp_prob}), assuming other parameters constant. Then, a valid choice of  ${U_p}\left(\boldsymbol{p};\boldsymbol{p^\nu}\right)$ is the first order approximation of  $U\left(\boldsymbol{p}\right)$, i.e.,  
\begin{equation}
{U_p}\left(\boldsymbol{p},\boldsymbol{p^\nu}\right)=\nabla_{\boldsymbol{p}}U\left(\boldsymbol{p^\nu}\right)^{T}\left(\boldsymbol{p}-\boldsymbol{p^\nu}\right)+\frac{\tau_{p}}{2}\left\Vert \boldsymbol{p}-\boldsymbol{b^\nu}\right\Vert ^{2}.\label{eq:U_b}
\end{equation}
where $\tau_p$ is a regularization parameter. The detailed steps can be seen in Algorithm \ref{alg:NOVA_Alg5}. Since all the constraints  have been shown to be convex in $\boldsymbol{p}$, the optimization problem in Step 1 of  Algorithm \ref{alg:NOVA_Alg5} can be solved by the standard projected gradient descent algorithm. 

%Similar to Access  Optimization problem, we get the minimizer for (\ref{eq:U_t_u_bar}), step 1 in Algorithm \ref{alg:NOVA_Alg1}), and then continue until convergence. 

%In this sub-problem, we solve for all $t_i$, for fixed placement and given accessing probabilities. Similar to $\boldsymbol{\pi}$- Optimization problem, we used NOVA algorithm to solve for $t$.  Algorithm \ref{alg:NOVA_Alg1}, shows the used procedure to solve for $\boldsymbol{t}$. After initialization with feasible choice, NOVA keeps generating  a sequence of monotonically decreasing $\left\{ \boldsymbol{t}^{\nu+1}\right\} $ until  convergence.

\begin{lemma} \label{lem_t}

For fixed $\boldsymbol{t}$, $\boldsymbol{w}$, $\boldsymbol{b}$, and  $\boldsymbol{q}$, the optimization of our problem over 
$\boldsymbol{p}$ generates a sequence of monotonically decreasing
objective values and therefore is guaranteed to converge to a stationary point.
\end{lemma}
%\begin{proof}
%Since NOVA Algorithm generates a sequence of monotonically decreasing
%objective values as shown in \cite{scutNOVA}, $\boldsymbol{t}$ is guaranteed to converge to a stationary point.
%\end{proof}

\subsubsection{Proposed Algorithm Convergence}

We first initialize 
$q_{i,j}^{(\ell)}$, $p_{j,\nu_j}^{(\ell)}$, $w_{j,\nu_j}$,  $t_i$ and $b_{i,\ell}$,  $\forall$ $i, j, \nu_j, \ell$ such that the choice is feasible for the problem. Then, we do alternating minimization over the five sub-problems defined above. Since each sub-problem converges (decreasing) and the overall problem is bounded from below, we have the following result.

%: chunk placement  $\mathcal{S}$ (Step 1), $\boldsymbol{\pi}$ (Step 2), and $\boldsymbol{t}$ (Step 3).   These steps are repeated until convergence. Note that in both Steps 2 and 3, we use NOVA algorithm while for the chunk placement, Hungarian algorithm is used to find the new placement. Since each sub-problem converges and monotonically decreasing and the overall problem is bounded from below, we have the following result.

%We are now ready to describe our algorithm, i.e., Algorithm \ref{alg:proAlg}. We first initialize $\mathcal{S}_{i}$, $\pi_{ij}$ and $t_i$ $\forall$ $i,j$ such that the choice is feasible for the problem. Then, we do alternating minimization over: $\boldsymbol{t}$ (Step 1 in our proposed algorithm), $\boldsymbol{\pi}$ (Step 2), and finally over chunk placement  $\mathcal{S}$ (Step 3). These steps are repeated until convergence. Note that in both Step 1 and Step 2, we use NOVA algorithm while for the chunk placement, Hungarian algorithm is used to find the new placement. 
%Since each sub-problem converges and monotonically decreasing, we get the following lemma.
\begin{theorem} 
 The proposed algorithm  converges to a  local optimal solution.
\end{theorem}

\section{Numerical Results}\label{sec:num}

In this section, we evaluate our proposed  algorithm for joint optimization of the mean stall duration and the average streamed video quality. 

%Further, we show the effect of the trade-off of the  parameter $\theta$. We study the two extremes where only either mean stall duration objective or video quality is considered. 

%Then, we show the tradeoff between the two QoE metrics based on the trade-off parameter $\theta$.

%how the impact of the tradeoff objective function varies with respect to the trade-off factor.  

{\scriptsize{}}
\begin{table}[b]
{\scriptsize{}\caption{{\small{}
The value of $\alpha_j/a_1$  used in the Numerical Results, where the units are 1/s.
\vspace{-.1in}
%For $1<\ell\leq10$, both $\beta_{j}^{(\ell)}$ and $\alpha_{j}^{(\ell)}$are scaled by $DR^{(\ell)}/DR^{(1)}$.
\label{tab:Storage-Nodes-Parameters}}}
}{\scriptsize \par}

{\scriptsize{}}%
\begin{tabular}{|c|>{\centering}p{2.53cc}|>{\centering}p{2.53cc}|c|c|c|}
\multicolumn{1}{>{\centering}p{2.53cc}}{{\scriptsize{}Node 1}} & \multicolumn{1}{>{\centering}p{2.53cc}}{{\scriptsize{}Node 2}} & \multicolumn{1}{>{\centering}p{2.53cc}}{{\scriptsize{}Node 3}} & \multicolumn{1}{c}{{\scriptsize{}Node 4}} & \multicolumn{1}{c}{{\scriptsize{}Node 5}} & \multicolumn{1}{c}{{\scriptsize{}Node 6}}\tabularnewline
\hline 
{\scriptsize{}$18.238$} & {\scriptsize{}$24.062$} & {\scriptsize{}$11.950$} & {\scriptsize{}$17.053$} & {\scriptsize{}$26.191$} & {\scriptsize{}$23.906$}\tabularnewline
\hline 
\end{tabular}{\scriptsize \par}

{\scriptsize{}}%
\begin{tabular}{|c|>{\centering}p{2.53cc}|>{\centering}p{2.53cc}|c|c|c|}
 \multicolumn{1}{>{\centering}p{2.53cc}}{{\scriptsize{}Node 7}} & \multicolumn{1}{>{\centering}p{2.53cc}}{{\scriptsize{}Node 8}} & \multicolumn{1}{>{\centering}p{2.53cc}}{{\scriptsize{}Node 9}} & \multicolumn{1}{c}{{\scriptsize{}Node 10}} & \multicolumn{1}{c}{{\scriptsize{}Node 11}} & \multicolumn{1}{c}{{\scriptsize{}Node 12}}\tabularnewline
\hline 
{\scriptsize{}$27.006$} & {\scriptsize{}$21.381$} & {\scriptsize{}$9.910$} & {\scriptsize{}$24.959$} & {\scriptsize{}$26.529$} & {\scriptsize{}$23.807$}\tabularnewline
\hline 
\end{tabular}{\scriptsize \par}
\vspace{-.2in}
\end{table}
{\scriptsize \par}

%\begin{table}[b]
%\caption{{\small{}Storage Node Parameters Used in our Simulation for $\ell=1$
%(Shift $\beta_{j}^{(1)}=10msec$ and rate $\alpha_{j}^{(1)}$ in 1/s).
%For $1<\ell\leq10$, both $\beta_{j}^{(\ell)}$ and $\alpha_{j}^{(\ell)}$
%are scaled by $DR^{(\ell)}/DR^{(1)}$.}}
%{\small{}}%
%\begin{tabular}{|c|c|c|c|c|c|c|}
%\multicolumn{1}{c}{} & \multicolumn{1}{c}{\textbf{\small{}Node 1}} & \multicolumn{1}{c}{\textbf{\small{}Node 2}} & \multicolumn{1}{c}{\textbf{\small{}Node 3}} & \multicolumn{1}{c}{\textbf{\small{}Node 4}} & \multicolumn{1}{c}{\textbf{\small{}Node 5}} & \multicolumn{1}{c}{\textbf{\small{}Node 6}}\tabularnewline
%\hline 
%{\small{}$\alpha_{j}^{(1)}$} & {\small{}$18.2298$} & {\small{}$24.0552$} & {\small{}$11.8750$} & {\small{}$17.0526$} & {\small{}$26.1912$} & {\small{}$23.9059$}\tabularnewline
%\hline 
%\end{tabular}{\small \par}
%
%{\small{}}%
%\begin{tabular}{|c|c|c|c|c|c|c|}
%\multicolumn{1}{c}{} & \multicolumn{1}{c}{\textbf{\small{}Node 7}} & \multicolumn{1}{c}{\textbf{\small{}Node 8}} & \multicolumn{1}{c}{\textbf{\small{}Node 9}} & \multicolumn{1}{c}{\textbf{\small{}Node 10}} & \multicolumn{1}{c}{\textbf{\small{}Node 11}} & \multicolumn{1}{c}{\textbf{\small{}Node 12}}\tabularnewline
%\hline 
%{\small{}$\alpha_{j}^{(1)}$} & {\small{}$27.006$} & {\small{}$21.3812$} & {\small{}$9.9106$} & {\small{}$24.9589$} & {\small{}$26.5288$} & {\small{}$21.8067$}\tabularnewline
%\hline 
%\end{tabular}\label{tab:Storage-Nodes-Parameters}
%\end{table}

\subsection{Parameter  Setup}
 We simulate our algorithm in a distributed storage system of $m=12$ distributed nodes, where each video file uses an $(7,4)$ erasure code. However, our model can be used for any given number of storage servers and for any erasure coding setting. We assume $d_j=20$ (unless otherwise explicitly stated) and $r=1000$ files, whose sizes are generated based on Pareto distribution \cite{arnold2015pareto} (as it is a commonly used distribution for file sizes \cite{Vaphase}) with shape factor of $2$ and scale of $300$, respectively. Since we assume that the video file sizes are not heavy-tailed, the first $1000$ file-sizes that are less than 60 minutes are chosen.  We also assume that the chunk service time  follows a shifted-exponential distribution with rate $\alpha_{j}^{(\ell)}$ and shift $\beta_{j}^{(\ell)}$, given as \eqref{eq:alphabeta}.  The value of $\beta_j a_1$ is chosen to be 10 ms, while the value of $\alpha_j/a_1$ is chosen as  in Table \ref{tab:Storage-Nodes-Parameters} (the parameters of $\alpha_j/a_1$ were chosen using a distribution, and kept fixed for the experiments). Unless explicitly stated, the arrival rate for the first $500$ files is $0.002s^{-1}$ while for the next $500$ files is set to be $0.003s^{-1}$. Chunk size $\tau$ is set to be equal to $4$ seconds (s). When generating video files, the size of each video file is rounded up to the multiple of $4$ seconds. The values of $a_\ell$ for the $4$ second chunk are given in Table \ref{tab:DataRates}, where the numbers have been taken from the dataset in  \cite{svcParm015}. %Thus, the different qualities $a_{\ell}$, $\ell = 1, \cdots, V$,  are $4s$ times the data rate in Table \ref{tab:DataRates}, where $V=6$. 
 We use a random placement of each file on  $7$ out of the $12$ servers.  In order to initialize our algorithm, we assume uniform scheduling,  $q_{i,j}^{(\ell)}=k/n$ on the placed servers and $p_{j,\nu_j}^{(\ell)} = 1/d_j$. Further, we choose  $t_{i}=0.01$,  $b_{i, \ell}=1/V$, and $w_{j, \nu_j}=1/d_j $. However, these choices of the initial parameters may not be feasible. Thus, we modify the parameter initialization
 % $\boldsymbol{q}$  
 to be closest norm feasible solutions.

% To get the size of the video in  bytes, we multiply the size of the video in seconds by the data rate (DR) of each quality level. Thus, the possible choices of $a_{\ell}$ are $4s$ 

%  and for $\ell \neq1$, we used (xx) and (xx) to calculate the values of $\alpha_j^{(\ell)}$ and $\beta_j^{(\ell)}$,  respectively. %, which are generated at random and kept fixed for the experiments.

 \begin{table}[t]
\caption{{\small{}Data Size (in $Mb$) of the different quality levels.\label{tab:DataRates}}}
\vspace{-.0in}
\centering
\begin{tabular}{|c|c|c|c|c|c|c|}
\hline 
$\ell$ & 1 & 2 & 3 & 4 & 5 & 6\tabularnewline
\hline 
$a_\ell$ & $6$ & $11$ & $19.2$ & $31.2$ & $41$ & $56.2$\tabularnewline
\hline 
\end{tabular}
\vspace{-.1in}
\end{table}
%  given above values of $\boldsymbol{p}$, $\boldsymbol{b}$ , $\boldsymbol{w}$ and $\boldsymbol{t}$.  

\subsection{Baselines}
We  compare our proposed approach with six strategies, which are described as follows. 
 \begin{enumerate}[leftmargin=0cm,itemindent=.5cm,labelwidth=\itemindent,labelsep=0cm,align=left]
 
% the placement is chosen at random where any $n$  out of $m$ servers are chosen for each file, where each choice is equally likely. Given the random placement,

%1
%\item {\em Proposed Algorithm:} In this strategy, the joint scheduler is determined by the proposed solution that
%minimizes the convex combination of the weighted average quality and the mean stall duration, with
%respect to the five sets of variables: auxiliary variables $\boldsymbol{t}$, server access probabilities $\boldsymbol{q}$, PS selection probabilities $\boldsymbol{p}$, bandwidth allocation weights $\boldsymbol{w}$, and video quality probabilities $\boldsymbol{b}$.
%% are alternatively optimized using the Algorithm in Section \ref{sec:algo}.

%2
\item {\em Projected Equal Access, Optimized Quality Probabilities, Auxiliary variables and Bandwidth Wights (PEA-QTB):} Starting with the initial solution mentioned above, the problem in \eqref{eq:joint_otp_prob} is optimized over the choice of $\boldsymbol{t}$, $\boldsymbol{b}$, $\boldsymbol{w}$, and $\boldsymbol{p}$ (using Algorithms \ref{alg:NOVA_Alg1}, \ref{alg:NOVA_Alg3},  \ref{alg:NOVA_Alg4}, and
\ref{alg:NOVA_Alg5}, 
respectively)  using alternating minimization. Thus, the value of $q_{i,j}^{(\ell)}$ will be approximately close to $k/n$ for the servers on which the content is placed, indicating equal access of the $k$-out-of-$n$ servers. 

%In this strategy, the initialization is as mentioned above

%we set $q_{i,j}^{(\ell)}=k/n$ $\forall i, \ell, j$ on the placed servers. Then, we modify the initialization of $\boldsymbol{q}$ to be closest norm feasible solution  given the values of $\boldsymbol{p}$, $\boldsymbol{t}$, $\boldsymbol{b}$ and $\boldsymbol{w}$. Finally, an alternating optimization over , $\boldsymbol{t}$, $\boldsymbol{b}$, $\boldsymbol{w}$ and $\boldsymbol{p}$ is performed to the objective using Algorithms \ref{alg:NOVA_Alg1}, \ref{alg:NOVA_Alg3},  \ref{alg:NOVA_Alg4}, and \ref{alg:NOVA_Alg5},  respectively.

%3

\item {\em Projected Equal Bandwidth, Optimized Quality Probabilities, Auxiliary variables and Server Access (PEB-QTA):} Starting with the initial solution mentioned above, the problem in \eqref{eq:joint_otp_prob} is optimized over the choice of $\boldsymbol{q}$, $\boldsymbol{t}$, $\boldsymbol{b}$, and $\boldsymbol{p}$ (using Algorithms \ref{alg:NOVA_Alg1Pi}, \ref{alg:NOVA_Alg1},  \ref{alg:NOVA_Alg3}, and \ref{alg:NOVA_Alg5}, respectively) using alternating minimization. Thus, the bandwidth split $w_{j,\nu_j}$ will be  approximately $1/d_j$. 

%In this strategy, we set $w_{j,\nu_j} = 1/d_j, \, \forall j, \nu_j$.  Then, we modify this choice to be closest norm feasible solution  given the values of other sub-problems optimization parameters. Finally, an alternating optimization over $\boldsymbol{q}$, $\boldsymbol{t}$, $\boldsymbol{b}$, and $\boldsymbol{p}$ is performed to the objective using Algorithms \ref{alg:NOVA_Alg1Pi}, \ref{alg:NOVA_Alg1},  \ref{alg:NOVA_Alg3}, and \ref{alg:NOVA_Alg5}, respectively.

%% 4
\item {\em Projected Equal Quality, Optimized Bandwidth Wights, Auxiliary variables and Server Access (PEQ-BTA):} Starting with the initial solution mentioned above, the problem in \eqref{eq:joint_otp_prob} is optimized over the choice of $\boldsymbol{q}$, $\boldsymbol{t}$, $\boldsymbol{w}$, and $ \boldsymbol{p}$ (using Algorithms  \ref{alg:NOVA_Alg1Pi}, \ref{alg:NOVA_Alg1},  \ref{alg:NOVA_Alg4}, and \ref{alg:NOVA_Alg5}, respectively)  using alternating minimization. Thus, the quality assignment, $b_{i,\ell}$ will be approximately $1/V$.

%In this strategy, we set $b_{i,\ell}=1/V,\, \forall i, \ell$ to be equal. Then, this choice is projected to be closest norm feasible solution. We then perform an alternating optimization over the remaining set of variables including $\boldsymbol{q}$, $\boldsymbol{t}$, $\boldsymbol{w}$, and $ \boldsymbol{p}$ using Algorithms  \ref{alg:NOVA_Alg1Pi}, \ref{alg:NOVA_Alg1},  \ref{alg:NOVA_Alg4}, and \ref{alg:NOVA_Alg5}, respectively.

%% 5
\item {\em Projected Proportional Service-Rate, Optimized Quality, Auxiliary variables and Bandwidth Wights (PSP-QTB):} In the initialization, the access probabilities among the servers on which file $i$ is placed, is given as  $q_{i,j}^{(\ell)}=k_{i}\frac{\mu_{j}^{(\ell)}}{\sum_{j}\mu_{j}^{(\ell)}}, \, \forall i,j,\ell $.
This policy assigns servers proportional to their service rates. The choice of all parameters are then modified to the closest norm feasible solution.  Using this initialization, the problem in \eqref{eq:joint_otp_prob} is optimized over the choice of $\boldsymbol{t}$, $\boldsymbol{b}$, $\boldsymbol{w}$, and $ \boldsymbol{p}$ (using Algorithms  \ref{alg:NOVA_Alg1},  \ref{alg:NOVA_Alg3},  \ref{alg:NOVA_Alg4}, and \ref{alg:NOVA_Alg5}, respectively)  using alternating minimization. 

%These access probabilities are projected toward feasible region to ensure
%stability of the storage system. With these fixed access probabilities,
%the joint quality-stall duration is optimized over the remaining four
%sets of variables: bandwidth allocation $\boldsymbol{w}$, quality
%probabilities $\boldsymbol{b}$, PSs selection probabilities $\boldsymbol{p}$, and the auxiliary variables $\boldsymbol{t}$.

%% 6
\item {\em Projected Lowest Quality, Optimized Bandwidth Wights, Auxiliary variables and Server Access (PLQ-BTA):}
In this strategy, we set
 $b_{i,1}=1\,\text{and \ensuremath{b_{i,\ell}=0},\ \ensuremath{\forall\ell}}\neq1$ in the initialization thus choosing the lowest quality for all videos.  Then, this choice is projected to the  closest norm feasible solution. Using this initialization, the problem in \eqref{eq:joint_otp_prob} is optimized over the choice of $\boldsymbol{q}$, $\boldsymbol{t}$, $\boldsymbol{w}$, and $ \boldsymbol{p}$  (using Algorithms  \ref{alg:NOVA_Alg1Pi}, \ref{alg:NOVA_Alg1},  \ref{alg:NOVA_Alg4}, and \ref{alg:NOVA_Alg5}, respectively)  using alternating minimization.

%% 7
\item {\em Projected Highest Quality, Optimized Bandwidth Wights, Auxiliary variables and Server Access (PHQ-BTA):}  In this strategy, we set
$b_{i,6}=1\,\text{and \ensuremath{b_{i,\ell}=0},\ \ensuremath{\forall\ell}}\neq6$ in the initialization thus choosing the highest quality for all videos.  Then, this choice is projected to the  closest norm feasible solution. Using this initialization, the problem in \eqref{eq:joint_otp_prob} is optimized over the choice of $\boldsymbol{q}$, $\boldsymbol{t}$, $\boldsymbol{w}$, and $ \boldsymbol{p}$  (using Algorithms  \ref{alg:NOVA_Alg1Pi}, \ref{alg:NOVA_Alg1},  \ref{alg:NOVA_Alg4}, and \ref{alg:NOVA_Alg5}, respectively)  using alternating minimization. 

\end{enumerate}

%\begin{figure*}[htb]
%	\centering
%	\begin{minipage}{.32\textwidth}
%		\centering
%		\includegraphics[trim=0in 0in 0in 0.0in, clip, width=\textwidth]{figs/stallconvrg}
%		%	\vspace{-.3in}
%		\captionof{figure}{Convergence of mean stall duration.}
%		\label{fig:ConvgMeanStall}
%	\end{minipage}%
%	\hspace{2mm}
%	\begin{minipage}{.32\textwidth}
%		\centering
%		\includegraphics[trim=0.0in 0in 0.0in 0in, clip, width=\textwidth]{figs/stall_vs_arrRate}
%		%	\vspace{-.3in}
%		\captionof{figure}{Mean stall duration for different video arrival rates with different video lengths. }
%		\label{fig:meanArrRateSameSize}
%	\end{minipage}
%	\hspace{2mm}
%	\begin{minipage}{.32\textwidth}
%		\centering
%		\includegraphics[trim=0.0in 0in 0.0in 0in, clip,width=\textwidth]{figs/qlty_vs_arrRate}
%		%		\vspace{-.3in}
%		\captionof{figure}{Mean stall duration for different video sizes considering  $\ell=3$.  }
%		\label{fig:meanArrRateDiffSize}
%	\end{minipage}
%	%\vspace{-.2in}
%	
%	\vspace{-.1in}
%\end{figure*}

\subsection{Results}
In this subsection, we
set $\theta = 10^{-7}$, {\em i.e.}, prioritizing stall minimization over quality enhancement. We note that the average quality numbers are orders of magnitude higher (since the quality term in \eqref{eq:joint_otp_prob} is proportional to the video length) than the mean stall duration and thus to bring the two to a comparable scale, the choice of $\theta = 10^{-7}$ is small.  This choice of $\theta$ is motivated since users prefer not seeing interruptions more than seeing better quality. In this section, we will consider the average quality definition as $\text{Average Quality } = \sum_{i,\ell}\frac{\lambda_i} {\overline{\lambda}} \frac{L_i}{\sum_{k=1}^r L_k} b_{i,\ell} a_{\ell}$. We note that the maximum average quality is bounded by $a_6 = 56.2$. The division by the sum of lengths is used as a normalization so that the numbers in the figures can be interpreted better.

\begin{figure}[t]
	\centering
	\includegraphics[trim=0in 0in 0in 0.0in, clip, width=.48\textwidth]{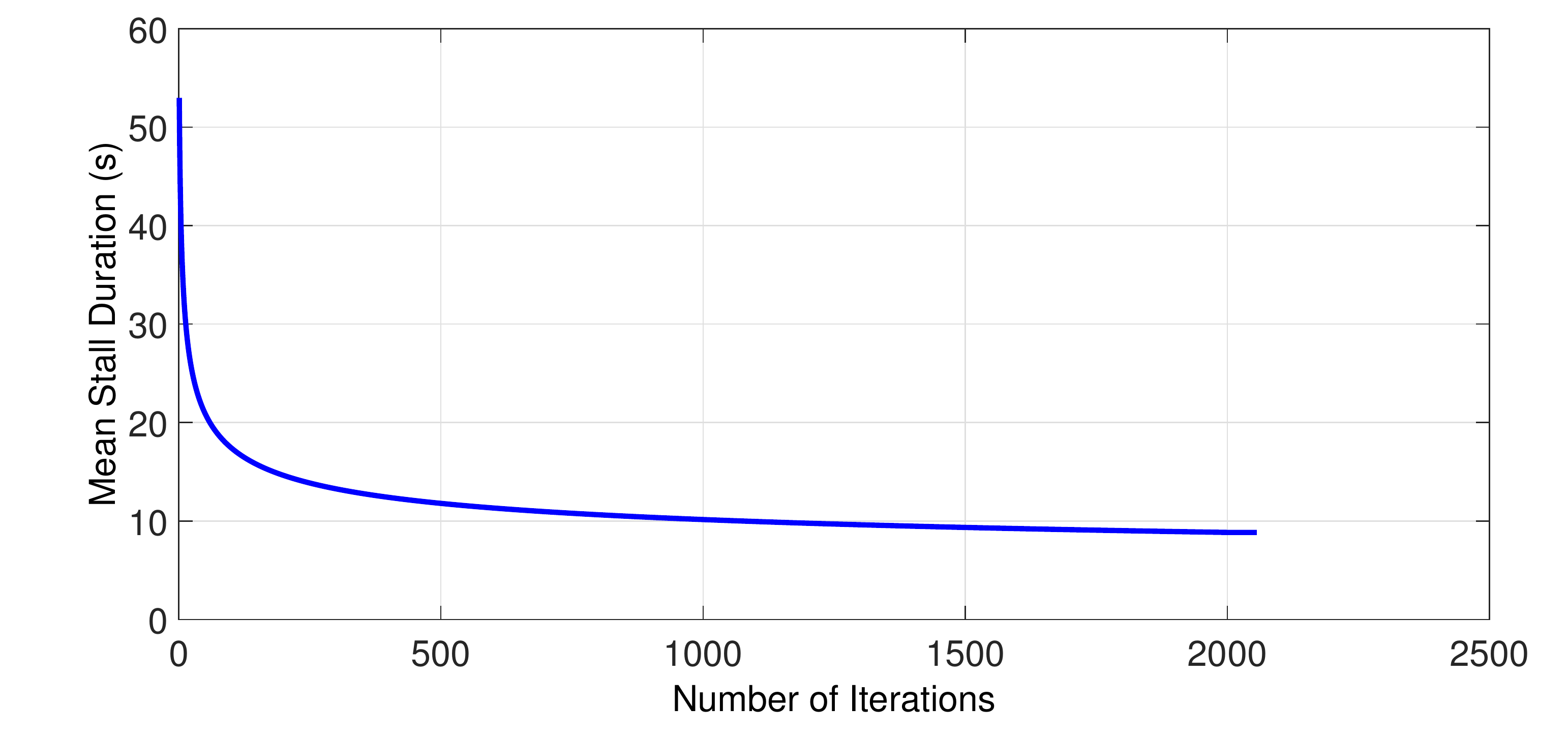}
			\vspace{-.16in}
		\captionof{figure}{Convergence of mean stall duration.}
		\label{fig:ConvgMeanStall}
		\vspace{-.2in}
\end{figure}

\begin{figure*}
	\centering
	\begin{minipage}{.32\textwidth}
		\centering
		\includegraphics[trim=0in 0in 4.1in 0in, clip, width=\textwidth]{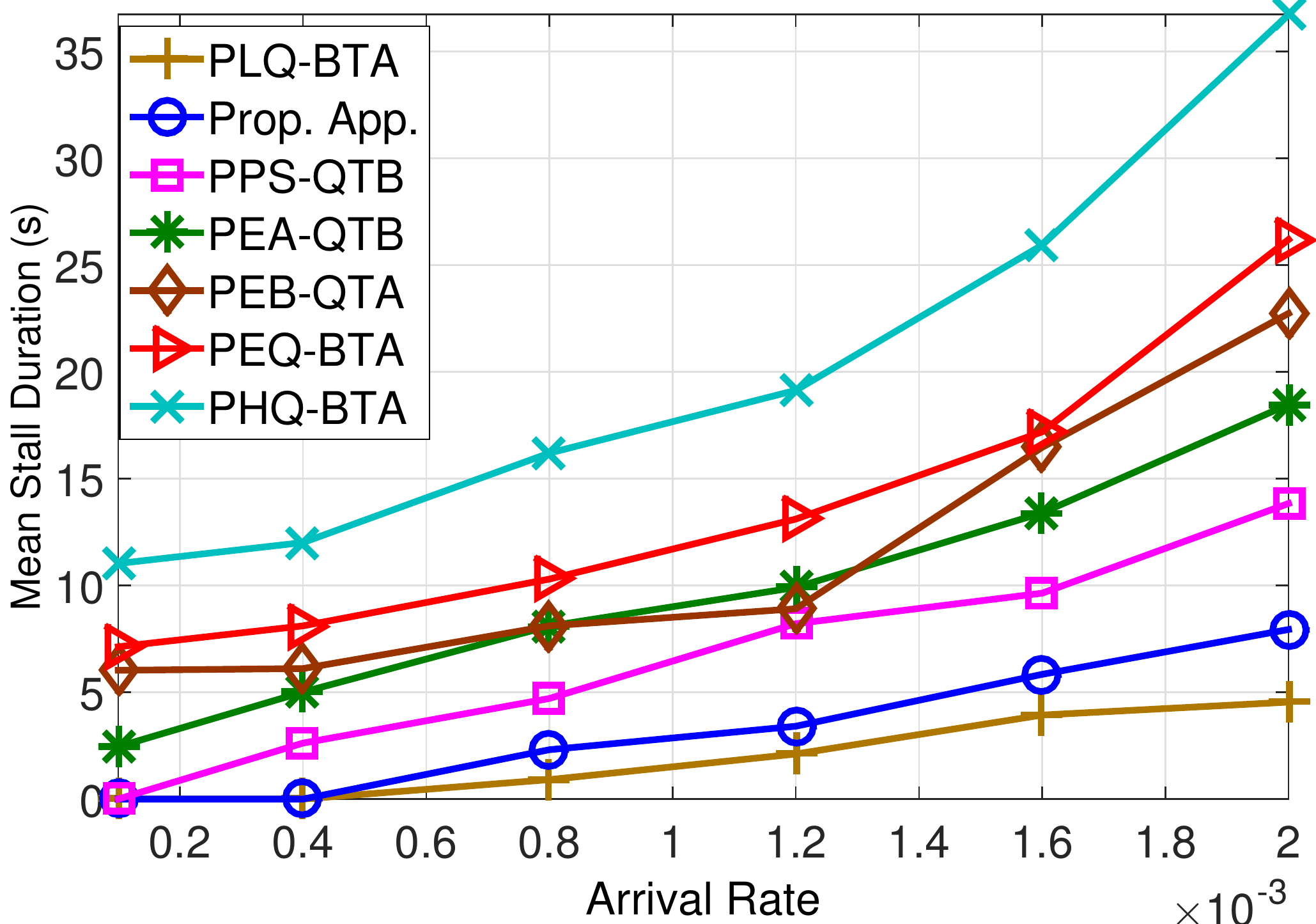}
%		\vspace{-.3in}
		\captionof{figure}{Mean stall duration for different video arrival rates.}
		\label{fig:meanStallVsArrRate}
	\end{minipage}%
	\hspace{2mm}
	\begin{minipage}{.31\textwidth}
		\centering
		\includegraphics[trim=0.1in 0in 4in 0in, clip, width=\textwidth]{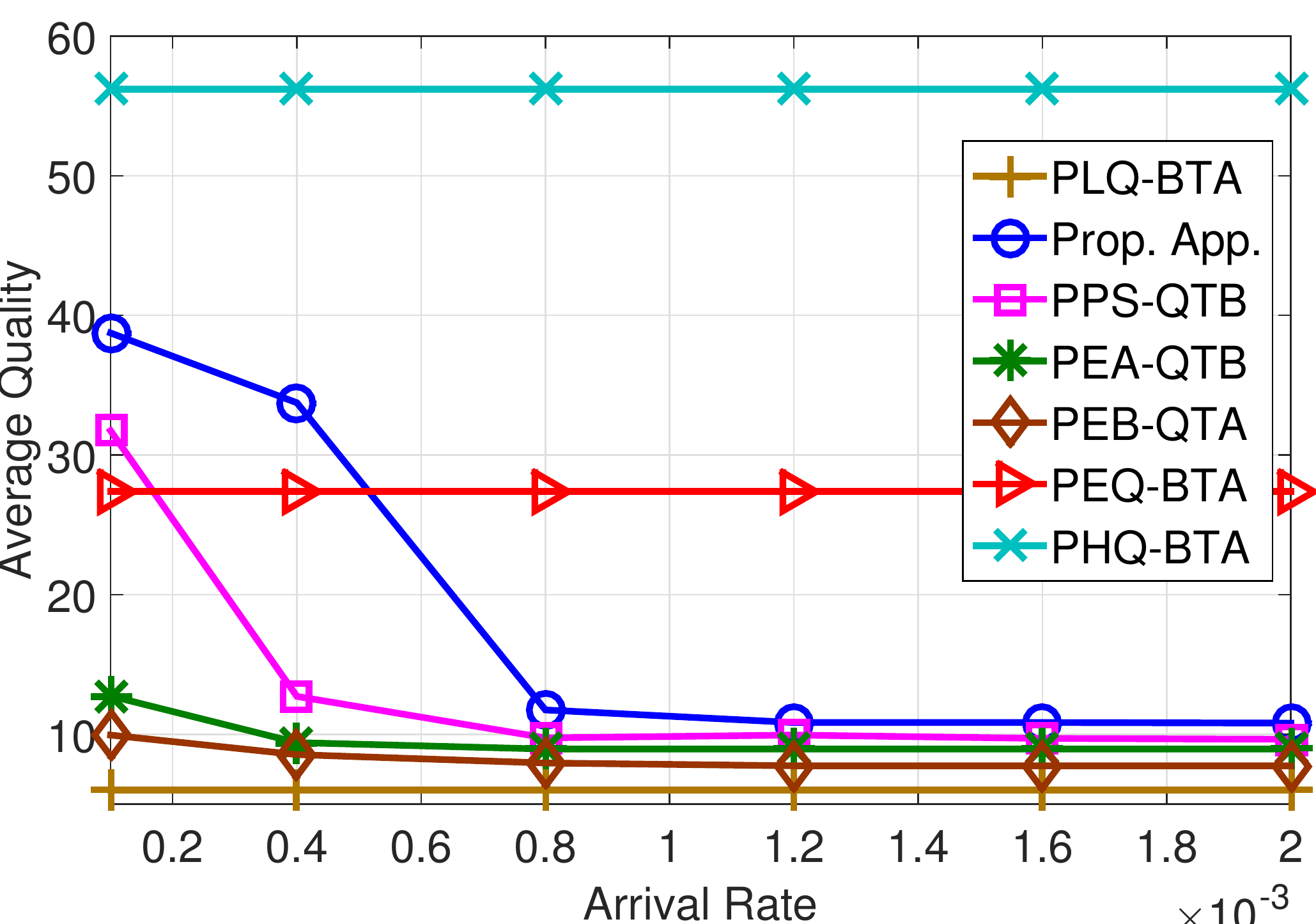}
%		\vspace{-.3in}
		\captionof{figure}{Average quality for different video arrival rates.}
		\label{fig:meanQltyVsArrRate}
	\end{minipage}
	\hspace{2mm}
	\begin{minipage}{.32\textwidth}
		\centering
		\includegraphics[trim=0in 0in 4in 0in, clip,width=\textwidth]{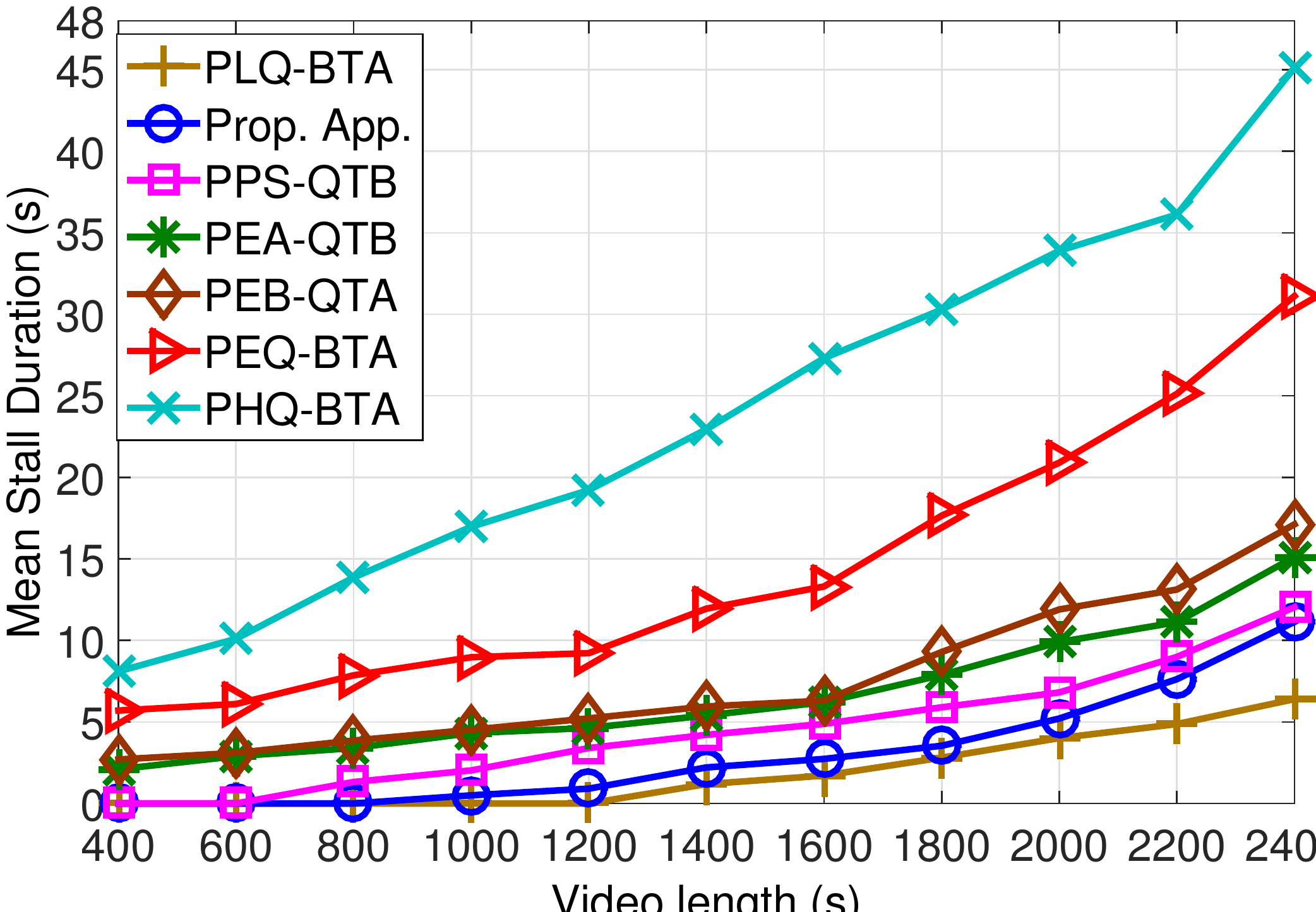}
%		\vspace{-.3in}
		\captionof{figure}{Mean stall duration for different video lengths.}
		\label{fig:meanStallVsVidSize} 
	\end{minipage}
%	\vspace{-.2in}
	\centering
	\begin{minipage}{.32\textwidth}
		\centering
		\includegraphics[trim=0.1in 0in 4.8in 0in, clip, width=\textwidth]{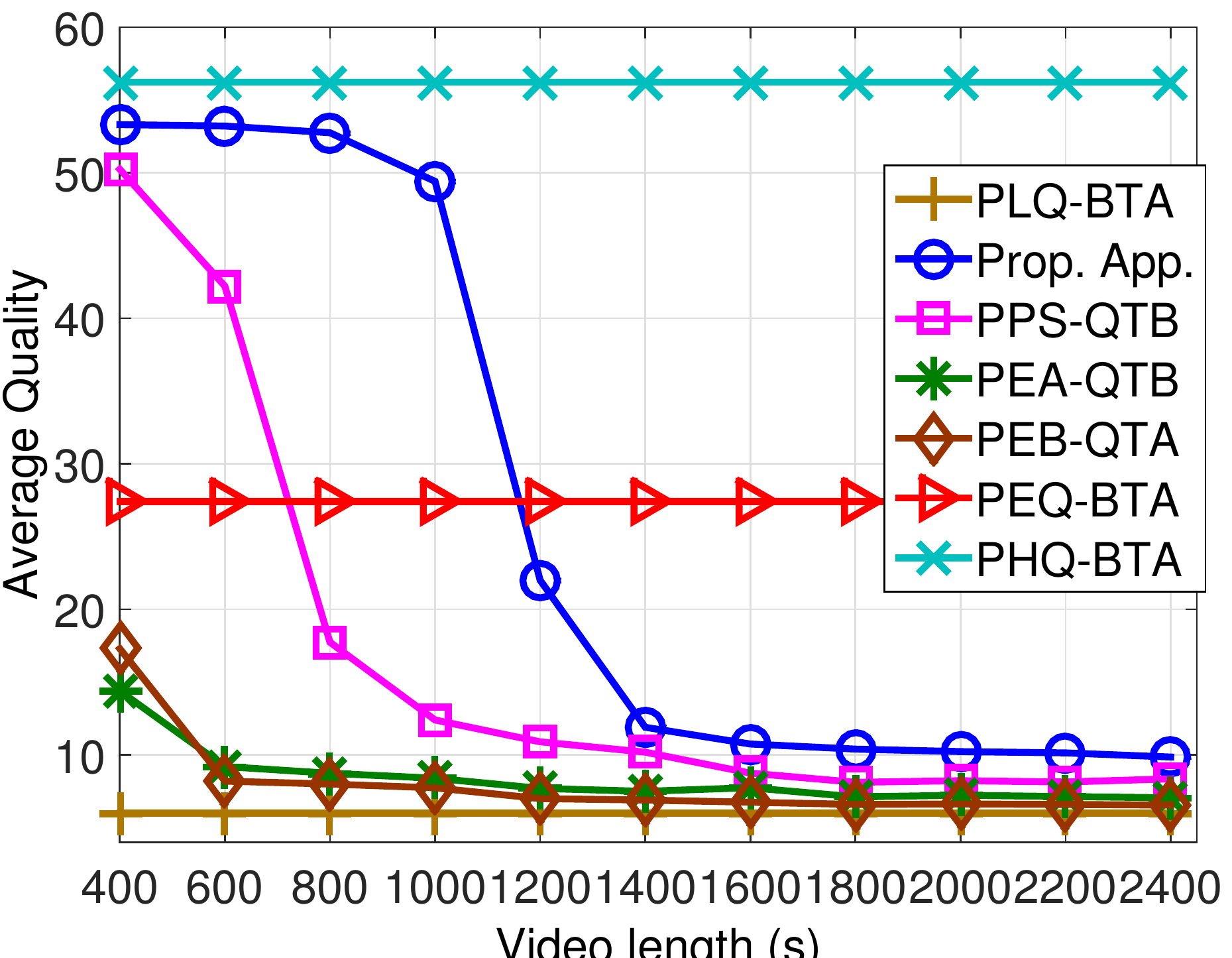}
%		\vspace{-.3in}
		\captionof{figure}{Average  quality for different video lengths.}
		\label{fig:meanQltyVsVidSize}
	\end{minipage}%
	\hspace{2mm}
	\begin{minipage}{.32\textwidth}
		\centering
		\includegraphics[trim=0.0in 0in 3.8in 0in, clip, width=\textwidth]{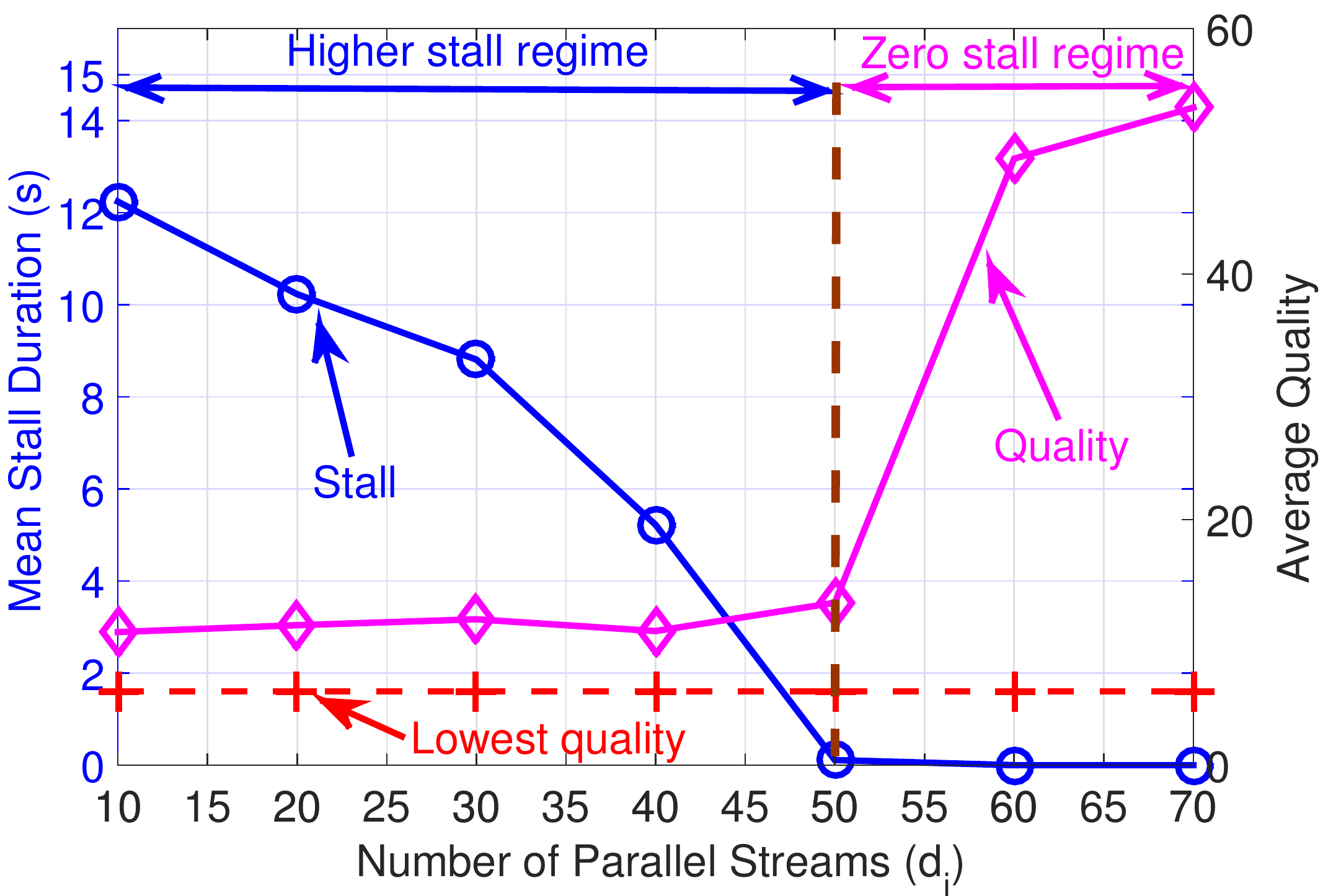}
%		\vspace{-.3in}
		\captionof{figure}{Average video quality and mean stall duration for different number of parallel streams $d_j$.}
		\label{fig:qlty_stall_Vs_d_j}
	\end{minipage}
	\hspace{2mm}
	\begin{minipage}{.32\textwidth}
		\centering
		\includegraphics[trim=0.1in 0.1in 4.4in 0.1in, clip,width=\textwidth]{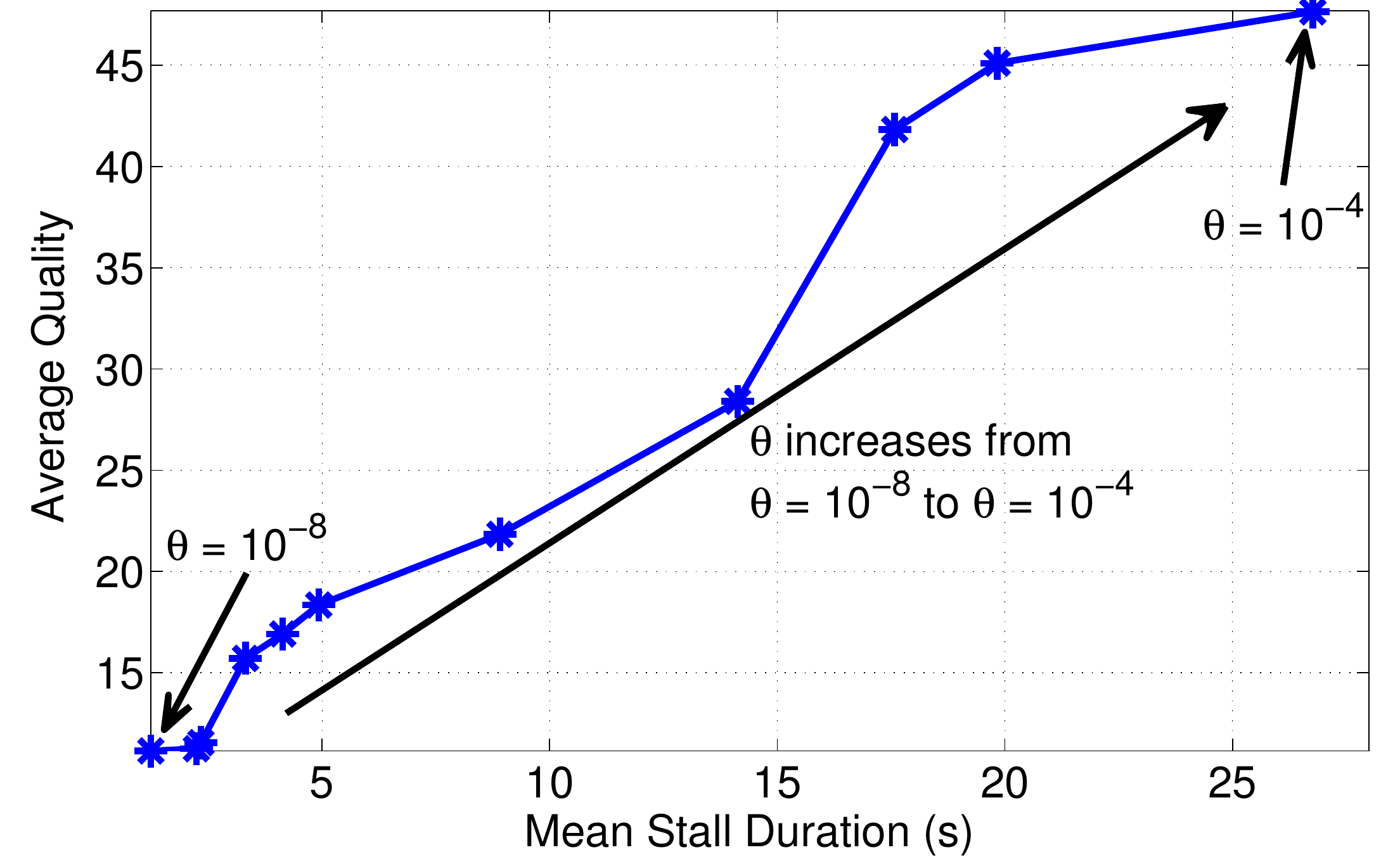}
%		\vspace{-.3in}
		\captionof{figure}{Tradeoff between mean stall duration and  average streamed video quality obtained by varying $\theta$.}
		\vspace{-.15in}
		\label{fig:tradeoff} 
	\end{minipage}
%	\vspace{-.2in}

\end{figure*}

%%%%%%%%%%%%%%%%%%%%%%%%%%%%%%%%%%%%%

%\begin{figure}[t].
%		\centering
%	\includegraphics[trim=0.1in 0in 0.0in 0in, clip, scale =0.35]{figs/stall_vs_ArrRate_all}
%			\vspace{-.2in}
%		\captionof{figure}{Mean stall duration for different video arrival rates. }
%		\label{fig:meanStallVsArrRate}
%		\vspace{-.18in}
%\end{figure}
%
%\begin{figure}[t]
%		\centering		\includegraphics[trim=0.1in 0in 0.0in 0.0in, clip,scale =0.35]{figs/qlty_vs_ArrRate_all}
%				\vspace{-.2in}
%		\captionof{figure}{Average quality for different video arrival rates.}
%		\label{fig:meanQltyVsArrRate}
%		\vspace{-.18in}
%\end{figure}

%\subsection{Mean Stall Duration optimization}
%In this subsection, we focus only on minimizing the mean stall duration of all files by setting $\theta=1$, {\em i.e.}, stall duration tail probability is not considered.% in this part of the simulation.

\subsubsection*{Convergence of the Proposed Algorithm}
Figure \ref{fig:ConvgMeanStall} shows the convergence  of our proposed algorithm, where we see the convergence of mean stall duration in about $2000$ iterations. 

% which alternatively optimizes the mean stall duration of all files over scheduling probabilities $\boldsymbol{\pi}$, auxiliary variables $\boldsymbol{t}$, quality probabilities $\boldsymbol{b}$, and bandwidth allocation weights $\boldsymbol{w}$. We notice that for $r=1000$ video files of size 600 sec with $m=12$ storage nodes and six quality levels shown in \ref{tab:DataRates}, the mean stall duration converges to the optimal value within less than $300$ iterations.  

\subsubsection*{Effect of Arrival Rate} We assume the arrival rate of all the files the same, and vary the arrival rates as depicted in Figures \ref{fig:meanStallVsArrRate} and \ref{fig:meanQltyVsArrRate}. These figures show the effect of different video arrival rates on the mean stall duration and  averaged quality, respectively. We note that PLQ-BTA achieves lowest stalls and lowest quality, since it fetches all videos at the lowest qualities. Similarly, PHQ-BTA has highest stalls, and highest video quality since it fetches all videos in the highest possible rate. The proposed algorithm has mean stall duration less than all the algorithms other than PLQ-BTA, and is very close to PLQ-BTA.  Further, the proposed algorithm has the highest video quality among all algorithms except PHQ-BTA and PEQ-BTA. Thus, the proposed algorithm helps optimize both the QoEs simultaneously achieving close to the best possible stall durations and achieving better average video quality than the baselines. With the choice of low $\theta$, the stall duration can be made very close to the stall duration achieved with the lowest quality while the proposed algorithm will still opportunistically increase quality of certain videos to obtain better average quality.

%We compare our proposed algorithm with the six baseline  policies and see that the proposed algorithm outperforms all baseline strategies for the QoE metrics of mean stall duration and  average quality. Thus, sever access, VMs selection and bandwidth weights allocation of files are important for the reduction of mean stall duration. While the mean stall duration increases with arrival, our approach still maintain low mean stall duration compared to the baselines. Moreover, we observe that the algorithms that optimize the quality tend to increase the quality of the streamed video only if the stall durations are zero (or close to zero) and adopt the lowest quality if the stall durations are not zero.
%Figure \ref{fig:meanQltyVsArrRate}
%shows that PLQ-BTA (PHQ-BTA) streaming policy realizes the lowest (highest) mean stall duration as it is aggressively minimize (maximize) the stall (quality) at the cost of other optimized parameters.

%Intuitively, having optimized placement leads to less stall duration (for all policies) since video chunks are balanced across the different storage nodes. Hence, the possibility of congested bottleneck nodes reduces accordingly. While mean stall duration increases as arrival rate increases, our approach, unlike PEAP and BNW, does not explode  as others do (especially for optimized one). 

\subsubsection*{Effect of Video Length}
The effect of having different video lengths on the mean stall duration and  average quality is also captured in Figures \ref{fig:meanStallVsVidSize} and \ref{fig:meanQltyVsVidSize}, respectively, where we assume that all the videos are of the same length. Apparently, the mean stall duration increases with the video length while the average quality decreases with the video length. The qualitative comparison of the different algorithms is the same as described in the case of varying arrival rates. Thus, at $\theta = 10^{-7}$, the proposed algorithm achieves the mean stall duration close to that of  PLQ-BTA while achieving significantly better quality. For algorithms other than PLQ-BTA, PEQ-BTA, and PHQ-BTA, the proposed algorithms outperforms all other baselines in both the metrics. %We also note that a fixed $\theta$ is used, and thus the average quality term gets more preference as the video lengths increase. Thus, for higher video lengths, 

\subsubsection*{ Effect of the Number of the Parallel Streams ($d_j$)}
Figure \ref{fig:qlty_stall_Vs_d_j} plots the average   streamed video quality and mean stall duration for varying number of parallel streams, $d_j$, for our proposed algorithm. We vary the number of PSs from $10$ to $70$ with increment step of 10 with $\theta=10^{-7}$. Increasing $d_j$ can only improve performance since some of the bandwidth splits can be zero thus giving the lower $d_j$ solution as one of the possible feasible solution. Increasing $d_j$ thus decreases stall durations by having more parallel streams, while increasing average quality. We note that for $d_j<50$, mean stall duration is non-zero and the stall duration decreases significantly while the average quality increases only slightly.  For $d_j>50$, the stall duration remains zero and the average video quality increases significantly with increase in $d_j$. Even though larger $d_j$ gives better results, the server may only be able to handle a limited parallel connections thus limiting the value of $d_j$ in the real systems.

\vspace{-.05in}
\subsubsection*{Tradeoff  between mean stall duration and  average video quality}

The preceding results show a trade off between the mean stall duration and
the  average quality of the streamed video. In order to investigate such tradeoff, Figure \ref{fig:tradeoff} plots the  average video quality versus the mean stall duration for different values of $\theta$ ranging from $\theta=10^{-8}$ to $\theta=10^{-4}$. This figure implies that a compromise between the two QoE metrics can be achieved by our proposed streaming algorithm by setting $\theta$  to an appropriate value. As expected, increasing $\theta$ will increase the mean stall duration as there is more priority to maximizing the average video quality.   Thus, an efficient tradeoff point between the QoE metrics can be chosen based on the service quality level desired by the service provider. %based on the point on the curve that is appropriate for the clients.
\subsection{Testbed Configuration and Implementation Results}

\begin{table}[t]
\caption{Testbed Configuration. \label{tab:Testbed-Configuration.}}

\centering%
\begin{tabular}{ccc}
\toprule 
\multicolumn{3}{c}{Cluster Information }\tabularnewline
\midrule
\midrule 
Control Plane & \multicolumn{2}{c}{OpenStack Kilo}\tabularnewline
\midrule 
VM Flavor & \multicolumn{2}{c}{1 VCPU, 2GB RAM, 20G storage (HDD)}\tabularnewline
\bottomrule
\end{tabular}\\
$\vphantom{}$\\
$\vphantom{}$

\centering%
\begin{tabular}{ccc}
\toprule 
\multicolumn{3}{c}{Software Configuration}\tabularnewline
\midrule
\midrule 
Operating System & \multicolumn{2}{c}{Ubuntu Server 16.04 LTS}\tabularnewline
\midrule 
Storage Server  & \multicolumn{2}{c}{Apache Server}\tabularnewline
\midrule 
Client & \multicolumn{2}{c}{Apache JMeter with HLS Sampler}\tabularnewline
\bottomrule
\end{tabular}

\end{table}
\begin{figure}
\centering
\includegraphics[trim=0.9in .5in 0.3in 0.6in, clip, width = .5\textwidth]{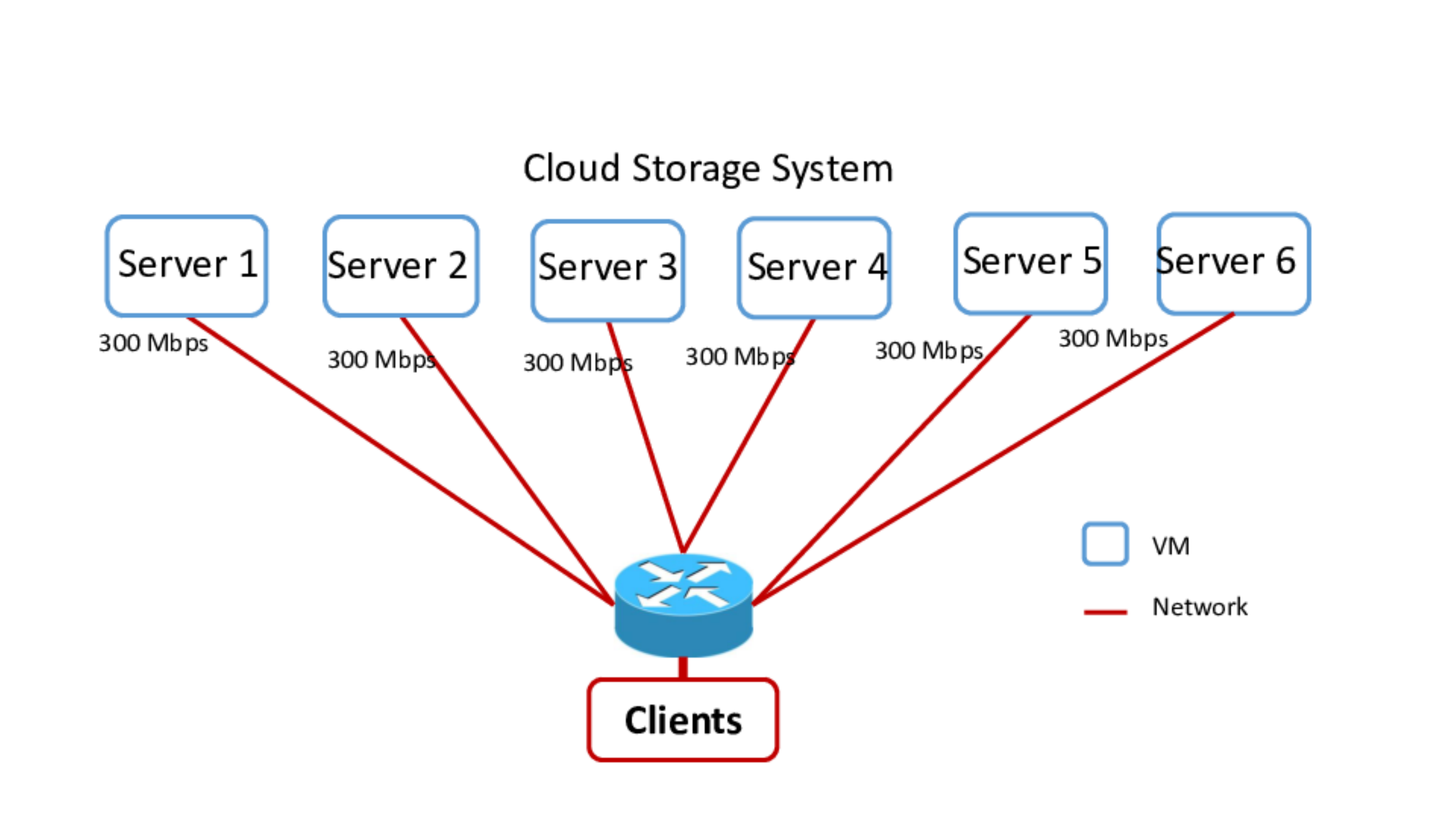}
\caption{Testbed in the cloud. \label{fig:Testbed-in-the}}
\end{figure}

An experimental environment in a virtualized cloud
environment is constructed. This  virtualized cloud is managed by open source software for creating private and
public cloud,  Openstack. We allocated 6 virtual machines
(VMs) as storage server nodes intended to store the
chunks. The schematic of our testbed is illustrated in Figure \ref{fig:Testbed-in-the}.
Table \ref{tab:Testbed-Configuration.} summarizes a detailed configuration
used for the experiments.

For client workload, we exploit a popular HTTP-trafic generator, Apache JMeter, with a plug-in that can generate traffic using HTTP Streaming protocol. 
We assume the amount of available bandwidth between origin server and each cache server is 200 Mbps, 
500 Mbps between cache server 1/2 and edge router 1, and 300 Mbps between cache server 3/4/5 and edge router 2. 
In this experiments, to allocate bandwidth to the clients, we throttle the client (i.e., JMeter) traffic according to the plan generated by our algorithm.  
We consider $500$ threads (i.e., users), $n=5$ , $k=3$ and set $e_j=40$, $d_j=20$. We chose the $(5,3)$ code as an example for our experiment. However, any other coding setting still works given that the required resources
are available. 
%Then, HLS sampler (i.e., request) is sent every $3$s. 
The video files are of length of $900$ seconds and the segment length is set to be $8$s. For each segment, we used JMeter built-in reports to estimate the downloaded time of each segment and then plug these times into our model to get the needed metric. 

%\begin{figure}[t]
%\centering\includegraphics[bb=0bp 0bp 600bp 407bp,clip,scale=0.35]{figsHW/wLatencyHW}
%
%\caption{Comparison of implementation results of our algorithm to analytical
%mean latency and PSP-based algorithm.\label{meanLatHW}}
%
%\end{figure}
\begin{figure}[t]
\centering\includegraphics[trim=0.0in 0in .6in .0in, clip,width=.65 \textwidth]{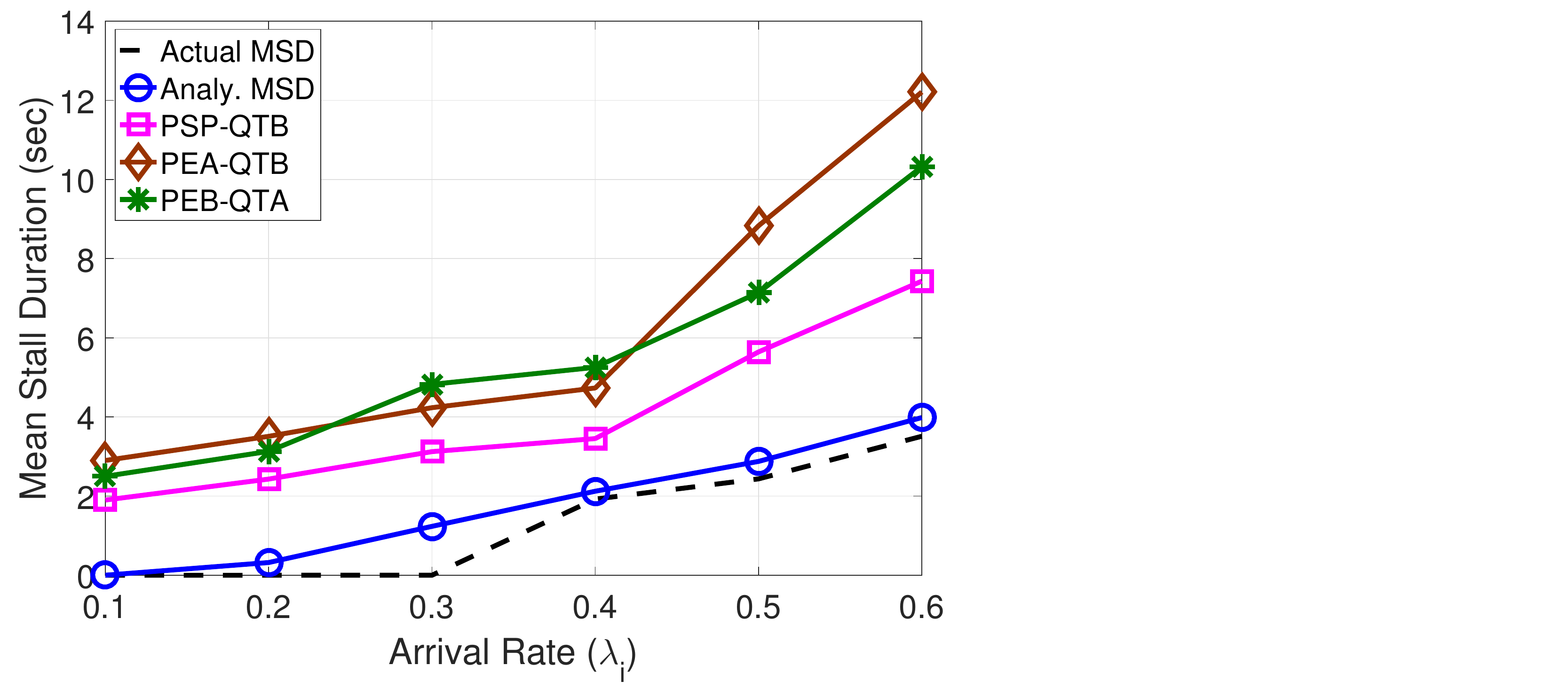}
\caption{Comparison of implementation results of our algorithm to analytical
mean stall durations, PSP-QTB, PEA-QTB, and PEB-QTA algorithms for different values of $\lambda_i$.\label{MSD}}

\end{figure}

Figure \ref{MSD} shows four different policies where we compare
the actual mean stall duration (MSD) for video files, analytical MSD,  PSP-QTB-based MSD, PEA-QTB-based MSD and PEB-QTA-based MSD  algorithms. We observe that the analytical MSD is very close to the actual measurements of the MSD obtained from our testbed, and approaches zero for reasonable large values of $\lambda_i$. Further, the proposed approach is shown to outperform the considered baselines. %To the best of our knowledge, this is the first
%work to jointly consider all key design degrees of freedom, including
%request (server) scheduling and the modeling variables associated
%with the tail latency bound. %Further, we can see that the gap between the analytical
%bound and the actual latency is very small and approaches zero for
%reasonable large values of $x$, i.e., $120$ s. 

\section{Conclusion}\label{sec:conc}

In this paper, a video streaming over cloud is considered where the content is erasure-coded on the distributed servers. We consider two quality of experience metrics to optimize: mean stall duration and  average quality of the streamed video. A two-stage probabilistic scheduling is proposed for the choice of servers and the parallel streams between the server and the edge router. Using the two-stage probabilistic scheduling and probabilistic quality assignment for the videos, an upper bound on the mean stall duration is derived. An optimization problem that minimizes a convex combination of the two QoE metrics is formulated, over the choice of two-stage probabilistic scheduling, probabilistic quality assignment, bandwidth allocation, and auxiliary variables.  Efficient algorithm is proposed to solve the optimization problem and the evaluation results depict the improved performance of the algorithm as compared to the considered baselines. 

%\bibliographystyle{plain}
%\bibliographystyle{ACM-Reference-Format}
%\pagebreak{}
\bibliographystyle{IEEEtran}

\bibliography{vidStallRef,allstorage,Tian,ref_Tian2,ref_Tian3,Vaneet_cloud,Tian_rest}
%\newpage
%\clearpage
\appendices
\newpage
\section{Algorithm Pseudo-codes for the Sub-problems}\label{apdx_table}

\begin{algorithm}[ht]
\caption{NOVA Algorithm to solve Access Optimization sub-problem\label{alg:NOVA_Alg1Pi}}

\begin{enumerate}
\item \textbf{Initialize} $\nu=0$, $k=0$,$\gamma^{\nu}\in\left(0,1\right]$,
$\epsilon>0$,$\boldsymbol{q}^{0}$ such that $\boldsymbol{q}^{0}$
is feasible ,
\item \textbf{while} $\mbox{obj}\left(k\right)-\mbox{obj}\left(k-1\right)\geq\epsilon$
\item $\quad$//\textit{\small{}Solve for $\boldsymbol{q}^{\nu+1}$ with
given $\boldsymbol{q}^{\nu}$}{\small \par}
\item $\quad$\textbf{Step 1}: Compute $\boldsymbol{\widehat{q}}\left(\boldsymbol{q}^{\nu}\right),$
the solution of $\boldsymbol{\widehat{q}}\left(\boldsymbol{q}^{\nu}\right)=$$\underset{\boldsymbol{q}}{\text{argmin}}$
$\boldsymbol{\widetilde{U}}\left(\boldsymbol{q},\boldsymbol{q}^{\nu}\right)\,\, $ s.t.
$\left(\ref{eq:rho_j}\right)$, $\left(\ref{eq:Lambda_j}\right)$, $\left(\ref{eq:sum_ij}\right)$, $\left(\ref{eq:pij}\right)$, 
$\left(\ref{eq:don_pos_cond2}\right)$, solved using  projected gradient descent 

\item $\quad$\textbf{Step 2}: $\ensuremath{\boldsymbol{q}^{\nu+1}=\boldsymbol{q}^{\nu}+\gamma^{\nu}\left(\widehat{\boldsymbol{q}}\left(\boldsymbol{q}^{\nu}\right)-\boldsymbol{q}^{\nu}\right)}$.
\item $\quad$//\textit{\small{}update index}{\small \par}
\item \textbf{Set} $\ensuremath{\nu\leftarrow\nu+1}$
\item \textbf{end while}
\item \textbf{output: }$\ensuremath{\widehat{\boldsymbol{q}}\left(\boldsymbol{q}^{\nu}\right)}$
\end{enumerate}
\end{algorithm}

\begin{algorithm}[ht]
	\caption{NOVA Algorithm to solve Auxiliary Variables  Optimization sub-problem\label{alg:NOVA_Alg1}}
	
	\begin{enumerate}
		\item \textbf{Initialize} $\nu=0$, $\gamma^{\nu}\in\left(0,1\right]$, $\epsilon>0$, $\boldsymbol{t}^{0}$ 
		such that $\boldsymbol{t}^{0}$  is feasible,
		\item \textbf{while} $\mbox{obj}\left(\nu\right)-\mbox{obj}\left(\nu-1\right)\geq\epsilon$
		\item $\quad$//\textit{\small{}Solve for  $\boldsymbol{t}^{\nu+1}$ with
			given $\boldsymbol{t}^{\nu}$}{\small \par}
		\item $\quad$\textbf{Step 1}: Compute $\boldsymbol{\widehat{t}}\left(\boldsymbol{t}^{\nu}\right),$
		the solution of  $\boldsymbol{\widehat{t}}\left(\boldsymbol{t}^{\nu}\right)=$$\underset{\boldsymbol{t}}{\text{argmin}}$
		$\boldsymbol{\overline{U}}\left(\boldsymbol{t},\boldsymbol{t}^{\nu}\right)$, s.t. \eqref{eq:t_i_alpha_j2}, \eqref{M_telda_less_1}, and  \eqref{eq:don_pos_cond2}
		using projected gradient descent
		\item $\quad$\textbf{Step 2}: $\ensuremath{\boldsymbol{t}^{\nu+1}=\boldsymbol{t}^{\nu}+\gamma^{\nu}\left(\widehat{\boldsymbol{t}}\left(\boldsymbol{t}^{\nu}\right)-\boldsymbol{t}^{\nu}\right)}$.
		\item $\quad$//\textit{\small{}update index}{\small \par}
		\item \textbf{Set} $\ensuremath{\nu\leftarrow\nu+1}$
		\item \textbf{end while}
		\item \textbf{output: }$\ensuremath{\widehat{\boldsymbol{t}}\left(\boldsymbol{t}^{\nu}\right)}$
	\end{enumerate}
\end{algorithm}
\begin{algorithm}[ht]
	\caption{NOVA Algorithm to solve Streamed Video Quality Optimization sub-problem\label{alg:NOVA_Alg3}}
	
	\begin{enumerate}
		\item \textbf{Initialize} $\nu=0$, $\gamma^{\nu}\in\left(0,1\right]$, $\epsilon>0$, $\boldsymbol{b}^{0}$ 
		such that $\boldsymbol{b}^{0}$  is feasible,
		\item \textbf{while} $\mbox{obj}\left(\nu\right)-\mbox{obj}\left(\nu-1\right)\geq\epsilon$
		\item $\quad$//\textit{\small{}Solve for  $\boldsymbol{b}^{\nu+1}$ with
			given $\boldsymbol{b}^{\nu}$}{\small \par}
		\item $\quad$\textbf{Step 1}: Compute $\boldsymbol{\widehat{b}}\left(\boldsymbol{b}^{\nu}\right),$
		the solution of  $\boldsymbol{\widehat{b}}\left(\boldsymbol{b}^{\nu}\right)=$$\underset{\boldsymbol{b}}{\text{argmin}}$
		$\boldsymbol{\overline{U}}\left(\boldsymbol{b},\boldsymbol{b}^{\nu}\right)$, s.t. \eqref{eq:rho_j},
		\eqref{eq:Lambda_j},
		\eqref{eq:sum_p_nuj},
		\eqref{eq:don_pos_cond2},
		using projected gradient descent
		\item $\quad$\textbf{Step 2}: $\ensuremath{\boldsymbol{b}^{\nu+1}=\boldsymbol{b}^{\nu}+\gamma^{\nu}\left(\widehat{\boldsymbol{b}}\left(\boldsymbol{b}^{\nu}\right)-\boldsymbol{b}^{\nu}\right)}$.
		\item $\quad$//\textit{\small{}update index}{\small \par}
		\item \textbf{Set} $\ensuremath{\nu\leftarrow\nu+1}$
		\item \textbf{end while}
		\item \textbf{output: }$\ensuremath{\widehat{\boldsymbol{b}}\left(\boldsymbol{b}^{\nu}\right)}$
	\end{enumerate}
\end{algorithm}

\begin{algorithm}[ht]
	\caption{NOVA Algorithm to solve Bandwidth Allocation Optimization sub-problem\label{alg:NOVA_Alg4}}
	
	\begin{enumerate}
		\item \textbf{Initialize} $\nu=0$, $\gamma^{\nu}\in\left(0,1\right]$, $\epsilon>0$, $\boldsymbol{w}^{0}$ 
		such that $\boldsymbol{w}^{0}$  is feasible,
		\item \textbf{while} $\mbox{obj}\left(\nu\right)-\mbox{obj}\left(\nu-1\right)\geq\epsilon$
		\item $\quad$//\textit{\small{}Solve for  $\boldsymbol{w}^{\nu+1}$ with
			given $\boldsymbol{w}^{\nu}$}{\small \par}
		\item $\quad$\textbf{Step 1}: Compute $\boldsymbol{\widehat{w}}\left(\boldsymbol{w}^{\nu}\right),$
		the solution of  $\boldsymbol{\widehat{w}}\left(\boldsymbol{w}^{\nu}\right)=$$\underset{\boldsymbol{w}}{\text{argmin}}$
		$\boldsymbol{\overline{U}}\left(\boldsymbol{w},\boldsymbol{w}^{\nu}\right)$, s.t. \eqref{eq:rho_j},
		\eqref{eq:w_alpha},
		\eqref{eq:sum_w_nuj}, and 
		\eqref{eq:don_pos_cond2}
		using projected gradient descent
		\item $\quad$\textbf{Step 2}: $\ensuremath{\boldsymbol{w}^{\nu+1}=\boldsymbol{w}^{\nu}+\gamma^{\nu}\left(\widehat{\boldsymbol{w}}\left(\boldsymbol{w}^{\nu}\right)-\boldsymbol{w}^{\nu}\right)}$.
		\item $\quad$//\textit{\small{}update index}{\small \par}
		\item \textbf{Set} $\ensuremath{\nu\leftarrow\nu+1}$
		\item \textbf{end while}
		\item \textbf{output: }$\ensuremath{\widehat{\boldsymbol{w}}\left(\boldsymbol{w}^{\nu}\right)}$
	\end{enumerate}
\end{algorithm}

\begin{algorithm}[ht]
	\caption{NOVA Algorithm to solve PS Selection Optimization sub-problem\label{alg:NOVA_Alg5}}
	
	\begin{enumerate}
		\item \textbf{Initialize} $\nu=0$, $\gamma^{\nu}\in\left(0,1\right]$, $\epsilon>0$, $\boldsymbol{p}^{0}$ 
		such that $\boldsymbol{p}^{0}$  is feasible,
		\item \textbf{while} $\mbox{obj}\left(\nu\right)-\mbox{obj}\left(\nu-1\right)\geq\epsilon$
		\item $\quad$//\textit{\small{}Solve for  $\boldsymbol{p}^{\nu+1}$ with
			given $\boldsymbol{p}^{\nu}$}{\small \par}
		\item $\quad$\textbf{Step 1}: Compute $\boldsymbol{\widehat{p}}\left(\boldsymbol{p}^{\nu}\right),$
		the solution of  $\boldsymbol{\widehat{p}}\left(\boldsymbol{p}^{\nu}\right)=$$\underset{\boldsymbol{p}}{\text{argmin}}$
		$\boldsymbol{\overline{U}}\left(\boldsymbol{p},\boldsymbol{p}^{\nu}\right)$, s.t. \eqref{eq:rho_j},
		\eqref{eq:Lambda_j},
		\eqref{eq:sum_p_nuj}, and
		\eqref{eq:don_pos_cond2},
		using projected gradient descent
		\item $\quad$\textbf{Step 2}: $\ensuremath{\boldsymbol{p}^{\nu+1}=\boldsymbol{p}^{\nu}+\gamma^{\nu}\left(\widehat{\boldsymbol{p}}\left(\boldsymbol{p}^{\nu}\right)-\boldsymbol{p}^{\nu}\right)}$.
		\item $\quad$//\textit{\small{}update index}{\small \par}
		\item \textbf{Set} $\ensuremath{\nu\leftarrow\nu+1}$
		\item \textbf{end while}
		\item \textbf{output: }$\ensuremath{\widehat{\boldsymbol{p}}\left(\boldsymbol{p}^{\nu}\right)}$
	\end{enumerate}
\end{algorithm}

\end{document}